\documentclass{article}
\usepackage[ruled]{algorithm2e}
\usepackage{microtype}%if unwanted, comment out or use option "draft"
\title{Deciding Maxmin Reachability in Half-Blind Stochastic
  Games\footnote{This work was partially supported by the French ANR project
    "Stoch-MC" and "LaBEX CPU" of Université de Bordeaux.}}
\author{Edon Kelmendi\\
            LaBRI\\
           Bordeaux, France\\
           \texttt{edon.kelmendi@labri.fr}
           \and
           Hugo Gimbert\\
           LaBRI \& CNRS\\
           Bordeaux, France\\
           \texttt{hugo.gimbert@labri.fr}}
\usepackage{amsmath,amsfonts,amssymb,amsthm,hyperref}
\usepackage{tikz,float,blkarray}
\usetikzlibrary{automata,positioning,shapes,calc}
\newcommand{\NN}{\mathbb{N}}

\newcommand{\PP}{\mathbb{P}}

\newcommand{\mon}{\mathcal{B}}
\newcommand{\emon}{\widetilde{\mathcal{B}}}
\newcommand{\mmon}{\mathcal{M}}
\newcommand{\emmon}{\widetilde{\mathcal{M}}}

\newcommand{\height}{\#-height}
\newcommand{\sq}[2]{\left ( #1 \right)_{#2}}

\newcommand{\hist}{\mathcal{H}}
\newcommand{\stma}{\mathbf{S_1}}
\newcommand{\stmi}{\mathbf{S_2}}
\newcommand{\sti}{\mathbf{S_i}}
\newcommand{\atma}{\mathbf{A_1}}
\newcommand{\atmi}{\mathbf{A_2}}
\newcommand{\ati}{\mathbf{A_i}}

\newcommand{\val}{val}

\newcommand{\bu}{\mathbf{u}}
\newcommand{\bv}{\mathbf{v}}
\newcommand{\bw}{\mathbf{w}}
\newcommand{\expr}{\mathfrak{E}}
\newcommand{\tldU}{\widetilde{U}}
\newcommand{\tldV}{\widetilde{V}}
\newcommand{\tldW}{\widetilde{W}}
\newcommand{\tldE}{\widetilde{E}}

\newtheorem{theorem}{\bf Theorem}
\newtheorem{definition}{\bf Definition}
\newtheorem{lemma}{\bf Lemma}
\newtheorem{problem}{\bf Problem}
\newtheorem{claim}{\bf Claim}
\begin{document}

\maketitle

\begin{abstract}
  Two-player, turn-based, stochastic games with reachability conditions are considered, where the maximizer has no
  information (he is blind) and is restricted to deterministic strategies whereas the minimizer is perfectly
  informed. We ask the question of whether the game has maxmin $1$, in other words we ask whether for all $\epsilon>0$
  there exists a deterministic strategy for the (blind) maximizer such that against all the strategies of the minimizer,
  it is possible to reach the set of final states with probability larger than $1-\epsilon$. This problem is undecidable
  in general, but we define a class of games, called leaktight half-blind games where the problem becomes decidable.  We
  also show that mixed strategies in general are stronger for both players and that optimal strategies for the minimizer
  might require infinite-memory.
\end{abstract}
\newpage
\tableofcontents
\newpage
\section{Introduction}
% !TEX root = asymmetric-limitsure.tex
Two-player stochastic games are a natural framework for modeling and
verification in the presence of uncertainty, where the problem of control is
reduced to the problem of optimal strategy synthesis \cite{AHK07}. There is a variety of
two-player stochastic games that have been studied, depending on the information
available to the players (perfect information or partial information), the
winning objective (safety, reachability, etc.), the winning condition (surely,
almost-surely, or limit-surely winning; probability higher than some quantity),
whether the players choose actions concurrently or whether they take turns.
 Stochastic games with partial observation are particularly well suited for
modeling many scenarios occurring in practice; normally we do not know the
exact state of the system we are trying to model, e.g. we are aided by noisy
sensors or by a software interface that provides only a partial
picture. Unfortunately, compared to perfect information games, algorithmic
problems on partial information games are substantially harder and often
undecidable \cite{chatterjee2014partial,PAZ66,GO10}. Assuming one player to be perfectly informed while the other
player is partially informed (semiperfect-information games \cite{CHAT05,CVD14})
brings some relief to the computational hardness as opposed to general partial
information games. 

In the present paper we consider half-blind
stochastic games: one player has no information (he is blind) and plays deterministically while the other
player is perfectly informed. We study half-blind games for the reachability
objective and maxmin winning condition: we want to decide if for every $\epsilon>0$
there exists a deterministic 
strategy for the maximizer such that against all strategies of the minimizer,
the final states are
reached with probability at least $1-\epsilon$.

The maxmin condition for half-blind games is a
generalization of the value $1$ problem for probabilistic finite automata \cite{RABIN63}.
Most decision problems on probabilistic finite automata are
undecidable, notably language emptiness \cite{PAZ66,bertoni1975solution,GO10},
and the value $1$ problem \cite{GO10}. Consequently, stochastic games with partial
information and quantitative winning conditions (the probability of fulfilling
the winning objective is larger than some quantity) are undecidable.
% The same is
%true for limit-sure winning conditions hence why in literature it is mostly the
%almost-sure winning condition that is studied. 
Nevertheless recently there has
been some effort on characterizing decidable classes of probabilistic automata
\cite{GO10,chatterjee2012decidable,chadha2009power,FGO12,leakoptimal}, with the leaktight
class~\cite{FGO12} subsuming the others~\cite{GKO15}.

\textbf{Our results.} In the present paper we show that a subclass of half-blind games called
leaktight games have a decidable maxmin reachability problem.
The game is abstracted through a finite algebraic
structure called the belief monoid. This is an extension to the Markov monoid
used in~\cite{FGO12}. Indeed the elements of the belief monoid are sets of
elements of the Markov monoid, and they contain information on the outcome of
the game when one strategy choice is fixed. The algorithm builds the belief
monoid and searches for particular elements which are witnesses
that the set of final states is maxmin reachable. The
proof of the correctness of the algorithm uses $k$-decomposition trees, a data
structure used in~\cite{COL13} that is related to Simon's factorization
forests. The $k$-decomposition trees are used to prove lower and upper bounds on
certain outcomes of the game and show that it behaves as predicted by the belief
monoid.

\textbf{Comparison with previous work. }The proof methods extends those developped
 in~\cite{FGO12} in three aspects. First, we define a new monoid structure
on top of the Markov monoid structure introduced in~\cite{FGO12}.
Second, 
we rely on the extension of Simon's
factorization forest theorem~\cite{SIM94} to $k$-factorization trees instead of $2$-factorization trees in~\cite{FGO12}
in order to derive 
upper and lower bound on the actual probabilities abstracted by the belief monoid.
Third, we rely on the leaktight hypothesis to prove both completeness and soundness,
while in the case of probabilistic automata the soundness of the abstraction by the Markov monoid was for free.

\textbf{Outline of the paper. } We start by fixing some notions and notation in
Section~\ref{sec:preliminaries} as well as providing a couple of examples. In
Section~\ref{sec:beliefmonoid} we introduce the belief monoid algorithm and the
Markov and belief monoids themselves. The $k$-decomposition tree data structure
used in the proofs of correctness is introduced in
Section~\ref{sec:kdecomptrees}, then in Section~\ref{sec:leak} the class of
leaktight games is defined using the notion of a leak. The correctness of the
algorithm is proved in Section~\ref{sec:correctness}, and finally we discuss
the power of different types of strategies in Section~\ref{sec:strategies} and
conclude.
%%% Local Variables:
%%% mode: latex
%%% TeX-master: "asymmetric-limitsure"
%%% End:

\section{Half-Blind Games and the Maxmin Reachability Problem}
\label{sec:preliminaries}
% !TEX root = asymmetric-limitsure.tex
Given a set $X$, we denote by $\Delta(X)$ the set of distributions on
$X$, i.e. functions $f\ :\ X\to [0,1]$ such that $\sum_{x\in X}f(x)=1$. 

A {\em half-blind} game is a two-player, zero-sum, stochastic, turn-based game,
played on a finite bipartite graph, where the maximizer has no information,
whereas the minimizer has perfect information. Formally a game $G$ is given by
the tuple $G=(\stma,\stmi,\atma,\atmi,p,F)$. The finite set $\sti$ is the states
controlled by Player $i$, the finite set $\ati$ is the actions available to
Player $i$ ($i=1,2$). Player 1 is the maximizer and Player 2 is the minimizer.
The function $p$ mapping $(\stma,\atma)$ to $\Delta(\stmi)$ and $(\stmi,\atmi)$
to $\Delta(\stma)$ gives the dynamics of the game. The sets $\stma,\stmi$ and
$\atma,\atmi$ are disjoint, i.e.  $\stma\cap\stmi = \emptyset$ and
$\atma\cap\atmi = \emptyset$. The set $F\subseteq \stma$ is the set of
final states.

A play of such a game takes place in turns. Initially the game is in some state
$s_1\in \stma$, then the maximizer (a.k.a. player 1) chooses some action
$a_1\in \atma$ which moves the game to some state $t_1\in \stmi$ selected
randomly according to the lottery $p(s_1,a_1)$. It is up to the minimizer
(a.k.a. player 2) now to choose some action $b_1\in \atmi$ which moves the game
to some state $s_2\in S_1$.  Then again maximizer chooses some action
$a_2\in \atma$ and so on, until the maximizer decides to stop, at which point,
if the game is in a state that belongs to the set of final states $F$, the
maximizer wins, otherwise it is the minimizer who wins. The maximizer is totally
blind and does not know what happens, he does not know in which state the game
is nor the actions played by minimizer.  Moreover the maximizer plays in a
deterministic way, he is not allowed to use a random generator to select his
actions.  As a consequence, the decisions of maximizer only depend on the time
elapsed and can be represented as words on $\atma$.  On the other hand, the
minimizer has full information and is allowed to plays actions selected
randomly.

Formally, the set of strategies for the maximizer is denoted by $\Sigma_1$ they consist of
finite words, i.e.  $\Sigma_1=\atma^*$. In order to emphasize that the
strategies of the maximizer are words, elements of $\Sigma_1$ are usually
denoted by $w$. 

The minimizer's strategies are functions from $\hist=(\stma\atma\stmi\atmi)^*\stma$ to
$\Delta(\atmi)$. Let $\Sigma_2$ be the set of such strategies. Its elements are
typically denoted by $\tau$.

%
%Formally, histories are words in the language
%$\hist=(\stma\atma\stmi\atmi)^*$. The length of a history $h\in\hist$,
%denoted $\len(h)$ is the number of turns, in other words, a history
%$h$ has length $k$ if $h$ is in the language
%$(\stma\atma\stmi\atmi)^k$. Denote $\hist_k$ the histories of
%length $k$, by $\hist^M$ 

Fixing strategies $w\in\Sigma_1$ of length $n$, $\tau\in\Sigma_2$ and an initial
state $s\in\stma$ gives a probability measure on the set $\hist_n=(\stma\atma\stmi\atmi)^n \stma$ which is
denoted by $\PP_s^{w,\tau}$: for a history
$h=s_1a_1t_1b_1\cdots s_na_nt_nb_ns_{n+1}\in \hist_n$,
\[
  \PP_s^{w,\tau}(h)=\prod_{i=1}^n p(s_i,a_i)(t_i)\cdot \tau(h_i)(b_i)\cdot p(t_i,b_i)(s_{i+1})
\]
if $s=s_1$ and $w=a_1\cdots a_n$, and $0$ otherwise, where
$h_i=s_1a_1t_1b_1\cdots s_ia_it_i$, $1\leq i\leq n$.

For $t\in \stma$, we will denote by $\PP_s^{w,\tau}(t)$ the chance of ending up
in state $t$ after starting from state $s$ and playing the respective
strategies, i.e
$\PP_s^{w,\tau}(t)=\sum_{ht\in\hist}\PP_s^{w,\tau}(ht)$. Whereas for a
set of states $R\subseteq \stma$ let
$\PP_s^{w,\tau}(R)=\sum_{t\in R}\PP_s^{w,\tau}(t)$.

\subsection{The Maxmin Reachability Problem}

Now we can introduce the maxmin reachability and for half-blind games, using
the notation and notions just defined.
Given a game with initial state $s\in\stma$ and final states
  $F\subseteq\stma$, the maxmin value 
  $\underline\val(s)$ is defined by
 \[
 \underline\val(s)=\sup_{w\in\Sigma_1}\inf_{\tau\in\Sigma_2}\PP_s^{w,\tau}(F)\enspace.
 \]
 In case $\underline\val(s)=1$, we say that $F$ is maxmin reachable from $s$.
 
\begin{problem}[Maxmin reachability]
  \label{prob:1}
  Given a game, is the set of final states $F$ maxmin reachable from the initial state $s$?
  \end{problem}

There is no hope to decide this problem in general.
The reason is that in the special case where the minimizer has no choice in any of the states that she controls, then
Problem~\ref{prob:1} is equivalent to the value one problem for {\em
  probabilistic finite automata} which is already known to be undecidable~\cite{GO10}. 
However, in the present paper, we establish that Problem 1 is decidable for a subclass of half-blind games called leaktight games.
%Similarly to \cite{FGO12}, we will identify the class of {\em
%  leaktight} half-blind games and show that --- given that the
%instance is leaktight --- we can decide whether $\underline\val(s)=1$.

\subsection{Deterministic Strategies for the Minimizer}

In general, strategies of the minimizer are functions from $\hist=(\stma\atma\stmi\atmi)^*\stma$ to
$\Delta(\atmi)$. However, because in the present paper we focus on the maxmin reachability problem,
we can assume that strategies of the minimizer have a  much simpler form:
the choice of action by the minimizer is deterministic and only depends on the current state and on how much time has elapsed since the beginning of the play.
Formally, we assume that minimizer strategies are functions $\NN \to (\stmi\to \atmi)$. 
Denote $\Sigma_2^p$ the set of all such strategies.
This restriction of the set of minimizer strategies does change the answer to the maxmin reachability problem because of the following theorem.

\begin{theorem}
  Given a game with initial state $s\in\stma$ and final states $F\subseteq\stma$
  we have
  \[
    \sup_{w\in\Sigma_1}\inf_{\tau\in\Sigma^p_2}\PP_s^{w,\tau}(F)=\sup_{w\in\Sigma_1}\inf_{\tau\in\Sigma_2}\PP_s^{w,\tau}(F). 
  \]
\end{theorem}
\begin{proof}
  Fixing a word $w\in\Sigma_1$ of length $n$, one can construct an MDP of finite
  horizon with state-space $\stmi\times \{1,\ldots, n\}$ and safety
  objective. Stationary strategies suffice to reach the safety objective here
  (see e.g. \cite{filar2012competitive}). A stationary strategy in this MDP
  is interpreted as a strategy in $\Sigma_2^p$ for the half-blind game. 
\end{proof}

\subsection{Two Examples}
The graph on which a half-blind game is played is visualized as in
Figures~\ref{fig:running_example2} and \ref{fig:running_example}. The circle
states are controlled by the maximizer, and the square states are controlled by
the minimizer, so for the example in Figure~\ref{fig:running_example2},
$\stma=\{i,f\}$ and $\stmi=\{1,2\}$. We represent only edges $(s,t)$ such that
$p(s,a)(t)>0$ for some action $a$ and we label the edge $(s,t)$ by $a$ if
$p(s,a)(t)=1$ and by $(a,p(s,a,t))$ otherwise.
\begin{figure}[H]
  \centering
 \begin{minipage}{.48\textwidth} 
  \begin{tikzpicture}[->,>=stealth,shorten >=1pt,auto,node distance=1.5cm,
    semithick, initial text=,font=\scriptsize]
    \tikzstyle{vertex1}=[circle,draw=black,minimum size=10pt,inner sep=0pt]
    \tikzstyle{vertex2}=[fill=black!25,draw=black,minimum size=10pt,inner sep=0pt]
    \node[vertex1,initial] (i) {$i$};
    \node[vertex2] (s1) [above = of i] {$1$};
    \node[vertex1,double] (f) [right of = i] {$f$};
    \node[vertex2] (s2) [below of = f] {$2$};
    \path

    (i) edge[bend right] node[right,pos=0.3] {\small $(a,\frac{1}{2})$} (s1)
    (i) edge[bend right] node[left,pos=0.3] {\small $(a,\frac{1}{2})$} (s2)
    (s1) edge[bend right] node[left] {$\alpha$} (i)
    (s1) edge[bend left] node {$\beta$} (f)
    (s2) edge[bend left] node {} (f)
    (f) edge[bend left] node {$a$} (s2);
  \end{tikzpicture}
  \\[.5cm]
  \caption{A half-blind game with $\underline\val(i)=1.$}
  \label{fig:running_example2}
\end{minipage}
\begin{minipage}{.48\textwidth} 
  \begin{tikzpicture}[->,>=stealth,shorten >=1pt,auto,node distance=1.5cm,
    semithick, initial text=,font=\scriptsize]
    \tikzstyle{vertex1}=[circle,draw=black,minimum size=10pt,inner sep=0pt]
    \tikzstyle{vertex2}=[fill=black!25,draw=black,minimum size=10pt,inner sep=0pt]
    \node[vertex1,initial] (i) {$i$};
    \node[vertex1] (s) [below = of i] {$s$};
    \node[vertex1] (c) [right = of i] {$c$};
    \node[vertex2] (t1) [above = of c] {$1$};
    \node[vertex2] (t2) [below = of c] {$2$};
    \node[vertex1,double] (f) [right = of c] {$f$};
    \node[vertex2] (t3) [above = of f] {$3$};
    \node[vertex2] (t4) [left = of s] {$4$};
    \path
    (i) edge[bend left] node[above left, pos=0.9] {\small $(a,\frac{1}{2})$} (t1)
    (i) edge[bend right] node[below left,pos=0.9] {\small $(a,\frac{1}{2})$} (t2)
    (i)  edge[bend right] node[right] {\small $b$} (t4)
    (t1) edge[bend left] node[above, pos=.7] {\small $\alpha_1$} (i)
    (t1) edge node[pos=.6] {\small $\alpha_2$} (c)
    (c) edge node[above left,pos=.4] {\small $a$} (t2)
    (c) edge node[pos=.7] {\small $b$} (t3)
    (t2) edge[bend right] node[above right, pos=.7] {\small $\beta$} (c)
    (t2) edge[bend right] node[left] {\small $(\alpha,\frac{3}{4})$} (i)
    (t2) edge[bend right] node[right,pos=.7] {\small $(\alpha,\frac{1}{4})$} (f)
    (f) edge[bend right] node[right] {\small $a,b$} (t3)
    (t3) edge[bend right] node {} (f)
    (s) edge[bend right] node[above right, pos=.8] {\small $a,b$} (t4)
    (t4) edge[bend right] node {} (s);
  \end{tikzpicture}
\caption{A half-blind game with $\underline\val(i)<1.$}
\label{fig:running_example}
\end{minipage}
\end{figure}

For the game in Figure~\ref{fig:running_example2} it is easy to see that
$\underline\val(i)=1$, since if the maximizer plays the strategy $a^n$, no
matter what strategy the minimizer chooses the probability to be on the final
state is at least $1-\frac{1}{2^n}$. On the other hand in the game depicted on
Figure~\ref{fig:running_example}, $\{f\}$ is not maxmin reachable from $i$.
 If the maximizer plays a strategy of only $a$'s then the minimizer always
plays the action $\beta$ and $\alpha_1$ for example and the probability to be in
the final state will be $0$. Therefore the maximizer has to play a $b$ at some
point. But then the strategy of the minimizer will be to play $\beta$ except
against the action just before $b$, against that action the minimizer plays
$\alpha$ letting at most $1/4$ of the chance to go to the final state, but
making sure that the rest of the probability distribution is stuck in the sink
state $s$. Consequently $\underline\val(s)=1/4$. It is interesting to note that
in the example in Figure~\ref{fig:running_example}, if we fix a strategy for the
minimizer first, then for all $\epsilon>0$ the maximizer can make the
probability of reaching the final state to be at least $1-\epsilon$ by playing
enough $a$'s to make sure that the token is either in $c$ or in $f$ and at that
point playing $b$, therefore $f$ is minmax reachable from $i$, but it is not
maxmin reachable. This is discussed in more details in Section~\ref{sec:strategies}.

We refer back to the game in Figure \ref{fig:running_example} in order to
illustrate the belief monoid algorithm in the next section.

%%% Local Variables:
%%% mode: latex
%%% TeX-master: "asymmetric-limitsure"
%%% End:

\section{The Belief Monoid Algorithm}
\label{sec:beliefmonoid}
% !TEX root = asymmetric-limitsure.tex
We abstract the game using two (finite) monoid structures that are constructed,
one on top of the other. Given that the game belongs to the class of leaktight
games, the monoids will contain enough information to decide maxmin reachability.

\subsection{The Markov Monoid} 
The Markov monoid is a finite algebraic object that is in fact richer than a
monoid; it is a {\em stabilisation} monoid (see \cite{colcombet2009theory}).
The Markov monoid was used in \cite{FGO12} to decide the value 1 problem for leaktight probabilistic
automata on finite words.
% We will use the Markov monoid to construct
%the {\em belief} monoid.

Elements of the Markov monoid are $\stma\times \stma$ binary matrices. They are
typically denoted by capital letters such as $U,V,W$. The entry that corresponds
to the states $s,t\in\stma$ is denoted by $U(s,t)$. We will make use of the
notation $s\xrightarrow{U} t$ in place of $U(s,t)=1$, when it is helpful.

We define two operations on
these matrices: the product and the iteration.
\begin{definition}
  Given two $\stma\times\stma$ binary matrices $U,V$, their {\em product} (denoted
  $UV$) is defined for all $s,t\in\stma$ as
  \[
    UV(s,t)=\begin{cases}1\ &\text{if } \exists s'\in\stma,\ s\xrightarrow{U}s'\wedge s'\xrightarrow{V}t=1,\\ 0\ &\text{otherwise.}\end{cases}
  \]
  Given a $\stma\times\stma$ binary matrix $U$ that is idempotent, i.e. $U^2=U$,
  its {\em iteration} (denoted $U^\#$) is defined for all $s,t\in\stma$ as
  \[
    U^\#(s,t)=\begin{cases}1\ &\text{if }s\xrightarrow{U}t\text{ and $t$ is
        $U$-recurrent,}\\ 0\ &\text{otherwise.}\end{cases}
  \]
  We say that some state $t\in \stma$ is $U$-recurrent, if for all $t'\in\stma$,
  $t\xrightarrow{U}t'\implies t'\xrightarrow{U} t$. Otherwise we say that $t$ is
  $U$-transient.
\end{definition}

For a set $X$ of binary matrices, we denote $\langle X \rangle$
the smallest set of binary matrices containing $X$ and closed under product and iteration.
Let $B^{a,\tau}$, $a\in \atma$, $\tau\in\Sigma_2^p$ be a matrix defined by
$s\xrightarrow{B^{a,\tau}}t\iff \PP_s^{a,\tau}(t)>0$, $s,t\in\stma$. Now the
definition of the Markov monoid can be given.

\begin{definition}[Markov monoid]
  The Markov monoid denoted $\mmon$ is
  \[
  \mmon = \left\langle \left\{B^{a,\tau}\ \mid\ a\in \atma, \tau\in\Sigma_2^p \right\} \cup \{\mathbf{1}\} \right\rangle\enspace,
  \]
  where $\mathbf{1}$ is the unit matrix. 
\end{definition}

\subsection{The Belief Monoid}
Roughly speaking, while the elements of the Markov monoid try to abstract the
outcome of the game when both strategies are fixed, the belief monoid tries to
abstract the {\em possible outcomes} of the game when only the strategy of the
maximizer is fixed. Hence the elements of the belief monoid are subsets of
$\mmon$, and they are typically denoted by boldfaced lowercase letters such as
$\bu,\bv,\bw$.

Given two elements of the belief monoid $\bu$ and $\bv$, their product is the
product of their elements, while the iteration of some idempotent $\bu$ is the
sub-Markov monoid that is generated by $\bu$ minus the elements in $\bu$ that
are not iterated. 

\begin{definition}
  Given $\bu,\bv\subseteq\mmon$, their product (denoted $\bu\bv$) is defined as
  \[
    \bu\bv=\{UV\ \mid\ U\in\bu, V\in\bv\}.
  \]
  Given $\bu\subseteq\mmon$ that is idempotent, i.e. $\bu^2=\bu$, its iteration
  (denoted $\bu^\#$) is defined as
  \[
    \bu^\#= \left\langle  \left\{    UE^\# V \mid  U,E,V\in \bu, EE=E  \right\} \right\rangle\enspace. 
  \]
\end{definition}
\newcommand{\bolda}{\mathbf{a}}
\newcommand{\limu}{\mathbf{u}}
\newcommand{\limv}{\mathbf{v}}
\newcommand{\boldeps}{\mathbf{\epsilon}}

Given $a\in \atma$, let $\mathbf{a}=\{B^{a,\tau} \mid \tau\in\Sigma_2^p\}$; we
give the definition of the belief monoid.
\begin{definition}[Belief Monoid]
  The belief monoid, denoted $\mon$, is the smallest subset of $2^\mmon$ that is
  closed under product and iteration and contains
  $\{\mathbf{a} \mid a\in \atma\}\cup\{\{\mathbf{1}\}\}$, where $\mathbf{1}$ is
  the unit matrix.
  % \marginpar{Hugo: dont you want to add the unit element so that it is a monoid?} yes
  %=\{\{B^{a,\tau}\ \mid\ \tau\in\Sigma_2^p\}\ \mid\ a\in \atma\}$.
\end{definition}

We are interested in a particular kind of elements in the belief monoid,
called \emph{reachability witnesses}.

\begin{definition}[Reachability Witness]
  An element $\bu\in\mon$ is called a reachability witness if for all $U\in\bu$,
  $s\xrightarrow{U} t\implies t\in F$, where $s$ is the initial state of the
  game and $F$ is the set of final states.
\end{definition}

%\paragraph*{Intuitive description of the abstraction} 
We give an informal
description of the way that the belief monoid abstracts the outcomes of the
game. Roughly speaking the strategy choice of the maximizer corresponds to
choosing an element $\bu\in\mon$ while the strategy choice of the minimizer
corresponds to picking some $U\in\bu$. Consequently under those strategy
choices, $U$ will tell us the outcome of the game, that is to say if for some
$s,t\in\stma$, if we have $s\xrightarrow{U} t$ then there is some positive
probability (larger than a uniform bound) of going from the state $s$ to the
state $t$. In case of $s\not\xrightarrow{U} t$ we will be ensured that the
probability of reaching the state $t$ from $s$ can be made arbitrarily small.
Therefore if a reachability witness is found then we will know that for any
strategy that the minimizer picks the probability of going to some non-final
state from the initial state can be made to be arbitrarily small.

\subsection{The Belief Monoid Algorithm}

\begin{algorithm}[h!t]
\caption{\label{algo:belief_monoid} The belief monoid algorithm.}
\SetAlgoLined
\KwData{A half-blind game.}
\KwResult{Answer to the Maxmin reachability problem.}
     
$\mon \gets \{\mathbf{a}\ \mid\ a\in \atma\}$.

Close $\mon$ by product and iteration\\

Return true iff there is a reachability witness in  $\mon$
\end{algorithm}

The belief monoid associated with a given game is computed by the 
belief monoid, see Algorithm~\ref{algo:belief_monoid}.
We will see later that under some condition, the belief monoid algorithm decides
the maxmin rechability problem.

We illustrate the computation of the belief monoid with an example.  Consider
the game represented on Figure~\ref{fig:running_example}. The minimizer has four
pure stationary strategies $\tau_{\alpha_1\alpha}$, mapping $1$ to $\alpha_1$
and $2$ to $\alpha$, and similarly the strategies
$\tau_{\alpha_1\beta},\tau_{\alpha_2\alpha},\tau_{\alpha_2\beta}$. Now we
compute $B^{a,\tau}$ where $\tau$ is one of the strategies above. Assume that we
have the following order on the states: $i<c<s<f$, then
$B^{a,\tau_{\alpha_1\alpha}}=\Biggl[\begin{smallmatrix}
  1 & 0 & 0 & 1\\
  1 & 0 & 0 & 1\\
  0 & 0 & 1 & 0\\
  0 & 0 & 0 & 1\\
      \end{smallmatrix}\Biggr]$, 
      $B^{a,\tau_{\alpha_1\beta}}=\Biggl[\begin{smallmatrix}
        1 & 1 & 0 & 0\\
        0 & 1 & 0 & 0\\
        0 & 0 & 1 & 0\\
        0 & 0 & 0 & 1\\
      \end{smallmatrix}\Biggr]$, 
      $ B^{a,\tau_{\alpha_2\alpha}}=
    \Biggl[\begin{smallmatrix}
        1 & 1 & 0 & 1\\
        1 & 0 & 0 & 1\\
        0 & 0 & 1 & 0\\
        0 & 0 & 0 & 1\\
        \end{smallmatrix}\Biggr]$, and 
        $B^{a,\tau_{\alpha_2\beta}}=
        \Biggl[\begin{smallmatrix}
        0 & 1 & 0 & 0\\
        0 & 1 & 0 & 1\\
        0 & 0 & 1 & 0\\
        0 & 0 & 0 & 1\\
      \end{smallmatrix}\Biggr]$. 
  The set that contains these matrices is the set $\mathbf{a}$. We can verify
  that $\mathbf{a}$ is not idempotent, since
  $U=B^{a,\tau_{\alpha_1\alpha}}B^{a,\tau_{\alpha_1\beta}}\not\in\mathbf{a}^2$, and the same
  for $V=B^{a,\tau_{\alpha_1\alpha}}B^{a,\tau_{\alpha_2\beta}}$. In fact
  $\mathbf{a}^2=\mathbf{a}\cup \{U,V\}$. The set $\mathbf{a}^2$ on the other
  hand is closed under taking products,
  i.e. $\mathbf{a}^4=\mathbf{a}^2$. Therefore we can take its iteration and
  compute the element $(\mathbf{a}^2)^\#$. The reader can verify that
  $(\mathbf{a}^2)^\#$ contains $(B^{a,\tau_{\alpha_1\alpha}})^\#=\Biggl[\begin{smallmatrix}
    0 & 0 & 0 & 1\\
    0 & 0 & 0 & 1\\
    0 & 0 & 1 & 0\\
    0 & 0 & 0 & 1\\
  \end{smallmatrix}\Biggr]
  $, $(B^{a,\tau_{\alpha_1\beta}})^\#=\Biggl[\begin{smallmatrix}
    0 & 1 & 0 & 0\\
    0 & 1 & 0 & 0\\
    0 & 0 & 1 & 0\\
    0 & 0 & 0 & 1\\
  \end{smallmatrix}\Biggr]$,
  $V$ and $B^{a,\tau_{\alpha_2\beta}}$. But it also contains
  $(B^{a,\tau_{\alpha_1\beta}})^\#B^{a,\tau_{\alpha_1\alpha}}=B^{a,\tau_{\alpha_1\alpha}}$. Therefore
  $(\mathbf{a}^2)^\#\mathbf{b}$ is not a reachability witness because if we pick
  $A=B^{a,\tau_{\alpha_1\alpha}}$ in $(\mathbf{a}^2)^\#$ and some
  $B\in \mathbf{b}$, we will have $i\xrightarrow{AB} s$, and $s$ is a sink
  state.

  This roughly tells us that maximizer cannot win with the strategies
  $\sq{(a^{2n}b)}{n}$, because against $a^{2n}b$ the minimizer plays the
  strategy $\tau_{\alpha_1\beta}$ for the first $2n-1$ turns and then plays the
  strategy $\tau_{\alpha_1\alpha}$ against the last $a$, making sure that after
  the $b$ is played the we end up in the sink state $s$ with at least $3/4$
  probability. Continuing the computation we can verify that the belief monoid
  of the game in Figure~\ref{fig:running_example} does not contain a
  reachability witness.

\subsection{The Extended Markov and Belief Monoids}

For defining leaktight half-blind games and in general for the
proofs of correctness of the belief monoid algorithm we use the {\em extended}
Markov and belief monoids. In simple words this means that we remember the
transitions which were deleted by the iteration operation. This extension is necessary 
for detecting leaks which will be defined in Section~\label{sec:leaks}.

The elements of the extended Markov monoid are pairs $(U,\tldU)$ of
$\stma\times\stma$ binary matrices where the right entry is not modified by the
iteration operation and stores the edges that were deleted from the left entry by 
the iteration operation.
Given two such pairs $(U,\tldU)$ and $(V,\tldV)$, define their product to be
$(U,\tldU)\cdot (V,\tldV)=(UV,\tldU\tldV)$. Given an idempotent $(E,\tldE)$,
define its iteration to be $(E,\tldE)^\#=(E^\#,\tldE)$.

\begin{definition}[Extended Markov Monoid]
  The extended Markov monoid (denoted $\emmon$) is the smallest set that is
  closed under product and iteration and contains
  $\{(B^{a,\tau},B^{a,\tau})\ \mid\ a\in\atma,
  \tau\in\Sigma_2^p\}\cup\{(\mathbf{1},\mathbf{1})\}$, where $\mathbf{1}$ is the
  unit matrix.
\end{definition}

The definition of the {\em extended} belief monoid (denoted $\emon$) remains the
same as that of the belief monoid except that its elements are now subsets of
$\emmon$.

%The right entries will keep track of transitions with non-zero probability but
%they will not be able to be uniformly bounded from below in general, i.e. for
%some $(U,\tldU)\in \bu$, $s\xrightarrow{\tldU} t$ will mean that under the fixed
%strategy choices the state $t$ can be reached from the state $s$ with non-zero
%probability, but we will not be able to ascertain that this probability is
%larger than a uniform bound. Whereas if $s\xrightarrow{U} t$, we will be able to
%bound the probability for going from state $s$ to $t$ from below with a uniform
%bound.

%To finish this section 
We give a few properties of the belief monoid that we
use in the sequel and leave their proofs as an exercise.

\begin{lemma} % \marginpar{do you want to state it for the belief or for the extended belief monoid?} Only used for \emon and \emmon
  \label{lem:properties}
  Let $\mathbf{e}\in\emon$ be an idempotent element of the extended belief
  monoid. Then the following hold: (1) $\emon$ together with the unit element
  $\{\{\mathbf{1}\}\}$ is a monoid; (2) $\mathbf{e}^\#$ is idempotent; (3)
  $(\mathbf{e}^\#)^\#=\mathbf{e}^\#$ and (4)
  $\mathbf{ee}^\#=\mathbf{e}^\#\mathbf{e}=\mathbf{e}^\#$.
\end{lemma}

The same properties also hold in the extended Markov monoid since it is a
stabilisation monoid \cite{FGO12}. 
%Properties (1)-(4) also trivially hold for the
%non-extended versions of both monoids.

%%% Local Variables:
%%% mode: latex
%%% TeX-master: "asymmetric-limitsure"
%%% End:

\section{$k$-Decomposition Trees}
\label{sec:kdecomptrees}
% !TEX root = asymmetric-limitsure.tex
The notion of $k$-decomposition trees was introduced in \cite{COL13}.
%In the present paper we will use this notion in two different contexts.
A $k$-decomposition tree is a data structure
for factorizing finite words into factors that are iterated with respect to some
finite monoid. In Section~\ref{sec:correctness} we will use a variant of Simon's factorization forest theorem
in order to bound the heights of $k$-decomposition trees,
which in turn will be used to obtain
upper and lower bounds on the probability of certain outcomes of the
game.

Let $A$ be a finite set, $(M,\cdot)$ a finite monoid and $\phi$ a morphism from
the free monoid of $A$ (i.e. $A^*$) to $M$. The set $A^*$ is infinite while $M$
is finite, so a pigeon-hole principle tells us that if we have a word $w$ that is
long enough it contains some factors $w_1,\ldots,w_n$ such that
$\phi(w_1)=\cdots = \phi(w_n)$. Simon's forest factorization theorem is a very strong extension
of this principle. It inductively factorizes the
factors themselves in a tree whose height is bounded by a function of the
size of the monoid independently of the length of the word $w$. Similarly to \cite{SIM94,COL13,FGO12} we
modify slightly this result to take into account the fact that $\mon$ and
$\mmon$ are not only monoids but they have some more structure.

First we define $k$-decomposition trees.
\begin{definition}[$k$-decomposition Tree]
  Let $A$ be a finite alphabet, $(M,\cdot)$ a finite monoid, equipped with a
  unary operation $\#$ that maps idempotents of $M$ to themselves:
  $\#:E(M)\to E(M)$ and $\phi$ a morphism from $A^*$ to $M$. The nodes of the
  $k$-decomposition tree are labeled by pairs $(u,U)$, where $u\in A^*$ and
  $U\in M$. The right entry of the pair is called the {\em type} of the
  node. Let $k>2$ and $w\in A^*$, then a $k$-decomposition tree of $w$ with respect to
  $M$ is a rooted and unranked tree whose root node is labeled by $(w,W)$ for
  some $W\in M$ and every node is one of the following kinds: (1) {\em leaves}
  do not contain any children and are labeled by $(a,\phi(a))$ for $a\in A$; (2)
  {\em product nodes} have exactly two children, the left one labeled by $(u,U)$
  and right one by $(v,V)$. The node itself is labeled by $(uv,UV)$; (3) {\em
    idempotent nodes} have at most $k-1$ children labeled by
  $(u_1,E),\ldots,(u_j,E)$ where $E\in E(M)$ is idempotent and $j<k$. The node
  itself is labeled by $(u_1\cdots u_j, E)$ and {\em iteration nodes} that have
  at least $k$ children labeled by $(u_1,E),\ldots, (u_j,E)$ where $E$ is
  idempotent and $j\geq k$. The node itself is labeled by
  $(u_1\cdots u_j,E^\#)$.
\end{definition}

The notion of a $k$-decomposition tree is introduced in \cite{COL13}, where it
is shown that for all $w\in A^*$ and $k>2$ there exists a $k$-decomposition tree
whose height depends only on the size of $M$ and not the length of the word ---
given that $M$ is a stabilisation monoid. We provide a similar proof, for
a slightly more general class of monoids that have the properties (1)-(4) given in Lemma~\ref{lem:properties}
whereas the definition of a stabilisation monoid requires extra axioms.
The proof was also given in \cite{FGO12} for the case $k=3$.

\begin{theorem}[\cite{SIM94,COL13,FGO12}]
  \label{theo:decomp height}
  Let $A$ be a finite alphabet, $(M,\cdot)$ a monoid equipped with a unary
  operator $\#$ that maps the idempotents of $M$ to themselves and has the
  properties (1)-(4) given in Lemma~\ref{lem:properties}, and $\phi$ a morphism
  from $A^*$ to $M$. For all $w\in A^*$, $k>3$ there exists a $k$-decomposition
  tree of $w$ with respect to $M$ whose height is at most $3\cdot |M|^2$.
\end{theorem}
We give a proof in the section that follows.

We will use $k$-decomposition trees in both the proof of soundness of the
belief monoid algorithm in Section~\ref{sec:soundness}, and its completeness in
Section~\ref{sec:completeness}. For soundness we construct $k$-decomposition
trees for words over the alphabet whose letters are pairs, where the left
component is a letter in $\atma$ and the right component is a stationary
strategy for the minimizer, with respect to the extended Markov monoid
$\emmon$. On the other hand for completeness we use $k$-decomposition trees over
the alphabet $\atma$ with respect to the monoid $\emon$. The $k$-decomposition
trees are used to prove lower and upper bounds on the probabilities of certain
outcomes.
\subsection{The Height of $k$-decomposition Trees}
\newcommand{\reld}{\mathcal{D}}
\newcommand{\relj}{\mathcal{J}}
\newcommand{\relh}{\mathcal{H}}
\newcommand{\rell}{\mathcal{L}}
\newcommand{\relr}{\mathcal{R}}

This section is devoted to proving Theorem~\ref{theo:decomp height}.

We start with Simon's factorization forest theorem. A {\em Ramseyan 
decomposition tree} is the same as a $k$-decomposition tree except that it does
not have iteration nodes, and there is no restriction on the number of children
of idempotent nodes. 

Let $A$ be a finite alphabet, $(M,\cdot)$ a finite monoid, and $\phi$
a morphism from $A^*$ to $M$. Then Simon's factorization forest theorem says:
\begin{theorem}[\cite{SIM94}]
  \label{theo:simon}
  For all $w\in A^*$ there exists a Ramseyan decomposition tree of $w$ with
  respect to $M$ whose height is at most $3\cdot |M|$.
\end{theorem}

Let $\#$ be a mapping from the idempotent elements of $M$ to themselves such
that the properties (1) through (4) in Lemma~\ref{lem:properties} hold. We will
prove Theorem~\ref{theo:decomp height}. Let $w\in A^*$ and $k>2$. We will prove
that there exists a $k$-decomposition tree of height at most
$3\cdot J \cdot |M|$, where $J$ is the number of
$\relj$-classes\footnote{$\relj$-classes are an important notion in
  the study of finite semi-groups and monoids. We give precise definitions
  below.} which is smaller than $|M|$.

According to Simon's factorization theorem there exists a Ramseyan decomposition
tree $T$ of $w$ of height at most $3\cdot |M|$.

Let $\phi_0=\phi$ and $A_0=A$.  

Call any idempotent node with children $(u_1,E),(u_2,E),\ldots ,(u_j,E)$, a {\em primitive iteration node} if
$E^\#\neq E$ and $j\geq k$.  If $T$ does not have any primitive iteration node, then $T_0=T$ itself is a
$k$-decomposition tree, and we are done. Otherwise {\em for all} primitive iteration nodes that are maximal in depth ---
i.e. there are no other primitive iteration nodes below --- labeled $(w,E)$ with children labeled
$(w_1,E),\ldots ,(w_j,E)$ where $w=w_1\cdots w_j$ and $j\geq k$, add a new letter of the alphabet $A_1=A_0\cup \{a_w\}$,
and change the morphism $\phi_1(a_w)=E^\#$ and $\phi_1(v)=\phi_0(v)$ for all other $v\in A_0^*$. The element $E^\#$ is
in the monoid $M$ since $E$ is idempotent. Also transform the word $u$ by replacing the factor $w$ by the letter $a_w$
and call this word $u_1$.

Now from Theorem \ref{theo:simon} applied to $M$ with alphabet $A_1$, morphism
$\phi_1$ and word $u_1$ there exists a Ramseyan decomposition tree $T_1$ of
height at most $3\cdot |M|$ where now the factor $w$ in $T_0$ is replaced by the
leaf $(a_w,E^\#)$. If $T_1$ does not contain any primitive iteration node then
we are done, we can {\em unwrap} the leaf $(a_w,E^\#)$ by replacing it with the
subtree of $T_0$ rooted in the primitive iteration node $(w,E)$, except that it
keeps the label $(w,E^\#)$. But if $T_1$ contains some primitive iteration node
then we recurse the process described above which returns an new alphabet $A_2$,
morphism $\phi_2$ and Ramseyan decomposition tree $T_2$.

Since we are removing more and more factors of the word $u$ and adding them as
new letters, repeating the procedure described above, must produce some $T_k$
that does not contain any primitive iteration nodes. We claim that
\begin{claim}
\label{claim:1}
$k\leq J$ where $J$ is the number of $\relj$-classes of $M$.
\end{claim}

So the number of times that we recurse the procedure above to transform a
Ramseyan decomposition tree to a $k$-decomposition tree whose height does not
depend on the length of the word \footnote{Notice that for this to be true at
  each step we have transform \em{all} the primitive iteration nodes of maximal
  depth and not one by one.}  but rather on the structure of $M$ itself. In fact
with Claim \ref{claim:1} the $k$-decomposition tree will have height at most
$3\cdot J\cdot |M|\leq 3\cdot |M|^2$.

To prove Claim \ref{claim:1} we need some results in the theory of
finite semigroups, in particular the Green's relations.

Let $U\in M$ an element of the monoid and define $UM=\{UU'\ \mid\ U'\in M\}$ and
$MUM=\{VUV'\ \mid\ V,V'\in M\}$.  Green's relations are four relations of
equivalence on the elements of $M$, denoted $\rell,\relr,\relj,\relh$ and
$\reld$ defined as follows. For a more detailed account of the Green's relations
and main theorems on finite semigroups see e.g. \cite{CLIF64,PIN86} etc.

\begin{definition}[Green's relations]
Let $U,V\in M$,
\begin{itemize}
\item $U \rell V \iff M U = M V,$
\item $U \relr V \iff UM = V M,$
\item $U \relj V \iff MUM =MVM,$
\item $U \relh V \iff U\rell V \text{ and }U\relr V,$ 
\item $U \reld V \iff \exists W \in M,U\relr W \text{ and }W\rell
  V\iff \exists W\in M,U\rell W \text{ and }W\relr V.$
\end{itemize}
\end{definition}
Where the last equivalence is because the relations $\relr$ and
$\rell$ commute.  Using these relations we can form partial orders
$\leq_\rell,\leq_\relr,\leq_\relj$, so that $U\leq_J V$ if and only if
$MUM\subseteq MVM$ and so on.

Observe that $UE^\#V\leq_\relj E^\#$, $UE^\#\leq_\relj E^\#$,
$E^\#V\leq_\relj E^\#$ for any two elements $U,V\in M$, so taking the product of
$E^\#$ with any other element, will produce another element of the monoid that
is smaller with respect to the relation $\leq_\relj$. Now we will show that for
any idempotent $E\in M$ if $E^\#\neq E$ then $E^\#<_\relj E$. This is Lemma 3 in
\cite{SIM94}.  Indeed the procedure above, when transforming primitive iteration
nodes, it replaces the label from $(w,E)$ to $(w,E^\#)$, hence the number of
times that this can be done is bounded by the number of $\relj$-classes hence
the Claim \ref{claim:1}.

Before we continue with the proof we need two lemmata from the theory of finite
monoids and semigroups.

\begin{lemma}
  \label{lem:uniqueidem}
  No $\relh$-class contains more than one idempotent element.
\end{lemma}
\begin{lemma}
  \label{lem:forjclass}
  Let $U,V\in M$,
  \begin{itemize}
  \item If $U\leq_\rell V$ and $U \relj V$ then $U\rell V$.
  \item If $U\leq_\relr V$ and $U \relj V$ then $U \relr V$.
  \end{itemize}
\end{lemma}

Proofs of Lemma \ref{lem:uniqueidem} and Lemma \ref{lem:forjclass}
can be found on any textbook on semigroup theory
e.g. \cite{CLIF64},\cite{PIN86}.

Now we are ready to prove that when iterating we descend the $\relj$-classes.

\begin{lemma}
  Let $E\in M$ an idempotent element such that $E^\#\neq E$. Then
  $E^\#<_\relj E$.
\end{lemma}
\begin{proof}
  Since $M$ fulfills the properties in Lemma~\ref{lem:properties}, in particular
  property (4), $EE^\#E=E^\#$ hence it follows that $E^\#\leq_\relj E$. We
  assume $E\relj E^\#$ and get a contradiction.  Regard that
  $E\leq_\rell E^\#=EE^\#$, therefore --- since $M$ is finite --- from Lemma
  \ref{lem:forjclass}, $E\rell E^\#$. The argument that $E\relr E^\#$ is
  dual. Consequently $E\relh E^\#$. Since both $E$ and $E^\#$ are idempotents in
  the same $\relh$-class, Lemma \ref{lem:uniqueidem} implies that $E=E^\#$ which
  is a contradiction.
\end{proof}

This concludes the proof of Theorem~\ref{theo:decomp height} and gives us a bound on the height of $k$-decomposition
trees that depends only on the size of the monoid $M$.

%%% Local Variables:
%%% mode: latex
%%% TeX-master: "asymmetric-limitsure"
%%% End:

\section{Leaks}
\label{sec:leak}
% !TEX root = asymmetric-limitsure.tex
Leaks were first introduced in \cite{FGO12} to define a decidable class of instances for the value $1$ problem
for probabilistic automata on finite words. The decidable class of 
\emph{leaktight automata} is general enough to encompass all known decidable classes for the value $1$ problem~\cite{GKO15}
and is optimal in some sense~\cite{leakoptimal}.
We extend the notion of leak from probabilistic automata to half-blind games
and prove that when a game does not contain any leak
then the belief monoid
algorithm decides the maxmin reachability problem.
%This notion seems to
%produce robust classes of decidable instances~\cite{GKO15}, and to capture well the hardness
%from which stems the undecidability~\cite{leakoptimal}. 

We illustrate leaks in the simplified case of probabilistic finite automata.
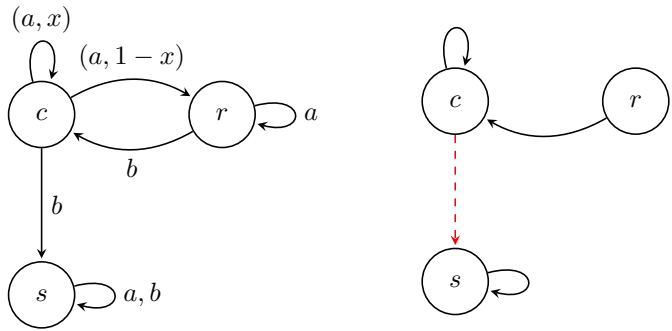
\begin{figure}[h!]
 \begin{center}
   \begin{minipage}{.45\textwidth}
   \begin{tikzpicture}[->,>=stealth,shorten >=1pt,auto,node distance=1.5cm,
     semithick, initial text=]
     \tikzstyle{vertex}=[circle,draw=black,minimum size=25pt,inner sep=0pt]
     \node[vertex] (c) {$c$};
     \node[vertex](s)[below=of c] {$s$};
     \node[vertex](r)[right = of c] {$r$};
     \path
     (c)  edge node {$b$} (s)
     (c) edge[loop above] node {$(a,x)$} ()
     (c) edge[bend left] node {$(a,1-x)$} (r)
     (r) edge[bend left] node {$b$} (c)
     (r) edge[loop right] node {$a$} ()
     (s) edge[loop right] node {$a,b$} ()
     ;
   \end{tikzpicture}
 \end{minipage}
     \begin{minipage}{.45\textwidth}
   \begin{tikzpicture}[->,>=stealth,shorten >=1pt,auto,node distance=1.5cm,
     semithick, initial text=]
     \tikzstyle{vertex}=[circle,draw=black,minimum size=25pt,inner sep=0pt]
     \node[vertex] (c) {$c$};
     \node[vertex](s)[below=of c] {$s$};
     \node[vertex](r)[right = of c] {$r$};
     \path
     (c) edge[loop above] node {} ()
     (c) edge[color=red,dashed] node {} (s)
     (r) edge[bend left] node {} (c)
     (s) edge[loop right] node {} ()
     ;
   \end{tikzpicture}
 \end{minipage}
\end{center}
\caption{A probabilistic finite automaton exhibiting a leak.}
\label{fig:leak}
\end{figure}

Probabilistic automata (PA) can be seen as the degenerate case of
half-blind games where the minimizer has no choice in any of the
states that she controls. Consider the PA (on the left) in
Figure~\ref{fig:leak}. When playing words from the sequence $\sq{a^{f(n)}}{n}$,
the probability of staying in state $c$ (if we start from state $c$) is
$x^{f(n)}$. Given that $0<x<1$ and that $f$ is an increasing function, we see
that this probability can be made arbitrarily small by choosing $n$ large
enough. Similarly playing words from the sequence $\sq{a^{f(n)}b}{n}$, starting
from the state $c$ the probability to go to the sink state $s$ is
$x^{f(n)}$. The question is what can we say about the outcome if we play words
from the sequence $\sq{(a^{f(n)}b)^{g(n)}}{n}$ for some increasing function
$g(n)$. For larger and larger $n$, is it the case that starting from the state
$c$ the probability of going to the sink state $s$ is bounded away from $1$? The
answer depends on the value of $x$ and the functions $f,g$. This behavior is
illustrated in Figure~\ref{fig:leak} on the right side. Each time $a^{f(n)}b$ is
played, the state $c$ {\em leaks} some probability to the sink state $s$,
denoted with the red dashed arrow. Having two or more leaks at the same time
complicates the matters further, and this is the difficulty making the
limit-sure decision problems undecidable in this setting.

Intuitively a leak happens when there is some communication between two recurrence classes with transitions that have a
small probability of occurring. Whether this small probability builds up to render one of the recurrence classes
transient is a computationally hard question to answer --- and in fact impossible in general. Other examples of leaks
can be found in~\cite{GKO15} and the link between leaks and convergence rates are discussed further
in~\cite{leakoptimal}.

We give a precise definition.

%In the instances where no leaks are present, the question of whether a set of
%states can be reached limit-surely does not depend on exact transition
%probabilities but only on whether the transition has some non-zero probability
%or not.

%Let us give a precise definition.
\begin{definition}[Leaks]
  An element of the extended Markov monoid $(U,\tldU)\in\emmon$ is a leak if it
  is idempotent and there exist $r,r'\in\stma$, such that: (1) $r,r'$ are
  $U$-recurrent, (2) $r\not\xrightarrow{U} r'$ and (3)
  $r\xrightarrow{\tldU} r'$. 

  An element of the extended belief monoid $\bu\in\emon$ is a leak if it
  contains $(U,\tldU)\in\bu$ such that $(U,\tldU)$ is a leak.

  A game is leaktight if its extended belief monoid does not contain any leaks.
\end{definition}

Note also that the question of whether a game is leaktight is decidable, since
this information can be found in the belief monoid itself.
%%% Local Variables:
%%% mode: latex
%%% TeX-master: "asymmetric-limitsure"
%%% End:

\section{Correctness of the Belief Monoid Algorithm\label{sec:correctness}}

This section contains the technical bulk of the paper since it is dedicated to proving that when the game is leaktight
the belief monoid algorithm is both sound (a reachability witness is found implies $\underline\val(s)=1$) and complete
(no reachability witness is found implies $\underline\val(s)<1$).

%The maxmin rechability problem (Problem 1) is decidable when the game is leaktight.

\begin{theorem}\label{theo:main}
The belief monoid algorithm
solves the maxmin reachability problem
for half-blind leaktight games.
\end{theorem}

Theorem~\ref{theo:main} is a direct consequence of Theorem~\ref{theo:soundness}
and Theorem~\ref{theo:completeness} which are given in the next two sections.

\subsection{Soundness}
\label{sec:soundness}
% !TEX root = asymmetric-limitsure.tex
In this section we give the main ideas to prove soundness of the belied monoid algorithm.
\begin{theorem}[Soundness]
\label{theo:soundness}
Assume that the game is leaktight and that its extended belief monoid contains a
reachability witness. Then the set of final
states is maxmin reachable from the initial state.
\end{theorem}

Theorem~\ref{theo:soundness} is justifying the {\em yes} instances of the belief
monoid algorithm, i.e. if the algorithm replies yes, then indeed
$\underline\val(s)=1$. It is interesting to note that the equivalent soundness
theorem for probabilistic automata in~\cite{FGO12} does not make use of the leaktight hypothesis. 
Theorem~\ref{theo:soundness} follows as a corollary of:
\begin{lemma}
  \label{lem:assocsub}
  Given a game whose extended belief monoid is leaktight, with every element $\bu\in\mon$
 of its belief monoid  we can associate a sequence $\sq{u_n}{n}$,
  $u_n\in \Sigma_1$ such that for all $\sq{\tau_n}{n}$, $\tau_n\in\Sigma_2^p$
  there exists $U\in\bu$ and a subsequence
  $\sq{(u_n',\tau_n')}{n}\subset \sq{(u_n,\tau_n)}{n}$ for which
  \[
    U(s,t)=0\implies \lim_n\PP_s^{u_n',\tau_n'}(t)=0,
  \]
  for all $s,t\in\stma$.
\end{lemma}

We can prove Theorem~\ref{theo:soundness} as follows. We are given a game that
is leaktight and has a reachability witness $\bu\in\mon$, to whom we can
associate a sequence of words $\sq{u_n}{n}$ according to
Lemma~\ref{lem:assocsub}. If on the contrary there exists $\epsilon>0$ such that
$\underline\val(s)\leq 1-\epsilon$ then there exists a sequence of strategies
$\sq{\tau_n}{n}$ such that for all $n\in\NN$,
$\PP_s^{u_n,\tau_n}(F)\leq 1-\epsilon'$, for some $\epsilon'>0$. This
contradicts Lemma~\ref{lem:assocsub} because for the reachability witness we
have by definition that for all $U\in\bu$, $U(s,t)=1$ implies $t\in F$.

We give a short sketch of the main ideas utilized into proving Lemma~\ref{lem:assocsub} before continuing with its proof
in the section that follows.

To $\mathbf{a}\in\mon$, $a\in\atma$ we associate the constant sequence of words
$\sq{a}{n}$. To the product of two elements in $\mon$ we associate the
concatenation of their respective sequences, and to $\bu^\#\in\mon$ the sequence
$\sq{u_n^n}{n}$ is associated, given that $\sq{u_n}{n}$ is coupled with
$\bu$. Then we consider words whose letters are pairs $(a,\tau)$, where
$a\in\atma$ and $\tau$ is a strategy that maps $\stmi$ to $\atmi$, i.e. a pure
and stationary strategy, and give a morphism from these words to the extended
Markov monoid $\emmon$. This allows us to construct $k$-decomposition trees of
such words with respect to $\emmon$. Then the $k$-decomposition trees are used
to prove lower and upper bounds on the outcomes of the game under the strategy
choices given by the word of pairs. The main idea is that we can construct for
longer and longer words, $k$-decomposition trees for larger and larger $k$,
thereby making sure that the iteration nodes have a large enough number of
children which enables us to show that the probability of being in transient
states is bounded above by a quantity that vanishes in the limit.
\subsubsection{Proof of Lemma~\ref{lem:assocsub}}
\newcommand{\typ}{\mathfrak{T}}
\newcommand{\fra}{\mathfrak{A}}
\newcommand{\pn}{\sq{p_n}{n}}
\newcommand{\pnp}{\sq{p_n'}{n}}
\newcommand{\squ}{\sq{u_n}{n}}
\newcommand{\sqv}{\sq{v_n}{n}}
\newcommand{\taun}{\sq{\tau_n}{n}}
\newcommand{\taupn}{\sq{\tau'_n}{n}}

Denote by $\Sigma_2'$ the set of pure and stationary strategies for the
minimizer, i.e. functions from $\stmi$ to $\atmi$. Let
\[
  \fra=\{(a,\alpha)\ \mid\ a\in\atma, \alpha\in\Sigma_2'\}.
\]

Note that $\fra$ is a finite set. Define $\phi$ the morphism from $\fra$ to
$\emmon$, that maps $(a,\alpha)$ to $(B^{a,\alpha},B^{a,\alpha})$. Since $a$ is
a single letter, taking stationary strategies is the same as taking strategies
from the set $\Sigma_2^p$ since what the strategy plays after the first turn
does not matter. Given a word $u\in \Sigma_1$ and a strategy
$\tau\in \Sigma_2^p$, the pair $(u,\tau)$ can be seen as a word over the
alphabet $\fra$. I.e. if $u=a_1\cdots a_n$ and $\tau(n)=\alpha_n$, $n\in\NN$, we
see $(u,\tau)$ as $(a_1,\alpha_1)(a_2,\alpha_2)\cdots (a_n,\alpha_n)$.

Given $p\in A^*$, $k>2$ and $h\in\NN$, let $T_k^h(p)$ be the set of all
$k$-decomposition trees of $w$ with respect to $\emmon$ whose height is at most
$h$. Denote by $\typ_k^h(w)\subseteq \emmon$ the set of types with which the
root nodes of the trees in $T_k^h(w)$ are labeled. Note that $T_k^h(w)$ and
consequently $\typ_k^h(w)$ can be empty, if $h$ is too small.

We define the notion of {\em reification}, which intuitively makes precise what
it means for a sequence of strategy choices (i.e. a sequence of words over the
alphabet $\fra$) to realize the abstraction that is provided by a subset of
$\emmon$. 
\begin{definition}[Reification]
  Let $\pn$ be a sequence of words over the alphabet $\fra$, $h\in\NN$ and
  $X\subseteq\emmon$. We say that $\pn$ {\em reifies} $X$ with height $h$ if
  there exists a subsequence $\pnp\subset \pn$ and $k\in\NN$ such that
  \[
    \typ_k^h(p'_n)=X, \text{ for }n\in\NN,
  \]
  moreover for infinitely many $i>k$, $X$ appears infinitely often in the
  sequence $\sq{\typ_i^h(p'_n)}{n}$.
\end{definition}

Reification is important because given that a sequence $\pn$ reifies some
$X\in\emmon$ with height $h$, we can prove lower and upper bounds on the
outcomes of the game under $\pn$ that {\em agree} with some element in $X$.

First we show that any sequence of words over the alphabet $\fra$ reifies some
$X\subseteq\emmon$.

\begin{lemma}
  Let $\pn$ be a sequence of words over the alphabet $\fra$. There exists
  $h\in\NN$ and $X\subseteq\emmon$ such that $\pn$ reifies $X$ with height $h$.
\end{lemma}
\begin{proof}
  Setting $h=3\cdot |\emmon|^2$, for all $k>2$ and $n\in\NN$, we have
  $\typ_k^h(\pn)\neq\emptyset$. This follows from Theorem~\ref{theo:decomp
    height}.

  Since $\emmon$ is finite there exists $X_1\subseteq\emmon$ and a subsequence
  $\pnp\subset \pn$ such that $\typ_3^h(p'_n)=X_1$ for all $n\in\NN$. If
  moreover there are infinitely many $i>3$ such that $X_1$ appears infinitely
  often in the sequence $\sq{\typ_i^h(p'_n)}{n}$, then the lemma
  concludes. Otherwise there exists some $k_1\in\NN$ such that for all $i>k_1$,
  $X_1$ appears only finitely often in the sequence
  $\sq{\typ_i^h(p'_n)}{n}$. Now choose some other subsequence
  $\sq{p_n''}{n}\subset\pnp$ and $X_2\subseteq\emmon$ such that
  $\typ_{k_1}^h(p''_n)=X_2$ for all $n\in\NN$. Since $\emmon$ is finite, the
  process above needs only a finite number of repetitions in order to find some
  $X\subseteq\emmon$ such that $\pn$ reifies $X$ with height $h$.
\end{proof}

Intuitively the lemma above means that under all sequences of strategy choices
for the players it is possible to find a subsequence under which the outcome of
the game is explained by an element of the extended Markov monoid. But this does
not say anything about the belief monoid. We would like for all $\bu\in\emon$ to
have a sequence of words $\squ$ over the alphabet $\atma$ such that for
any sequence of strategies $\taun$, $\pn=\sq{(u_n,\tau_n)}{n}$ reifies
some $X\subseteq \emmon$ and moreover $X$ and $\bu$ have at least one element in
common. This is the purpose of the next lemma.

\begin{lemma}
  Let $\bu\in\emon$, then there exists a sequence of words $\squ$ over
  the alphabet $\atma$, $h\in\NN$ and a function $N : \NN\to\NN$, such that for
  all sequence of strategies $\taun$ in $\Sigma_2^p$, $k>2$ and
  $n>N(k)$,
  \[
    \typ_k^h\big((u_n,\tau_n)\big)\cap \bu\neq\emptyset.
  \]
\end{lemma}
\begin{proof}
  We proceed by induction on the elements of $\emon$.
  \begin{itemize}
  \item \textbf{Base case.} For elements $\mathbf{a}\in\emon$, where
    $a\in\atma$, set the sequence of words to be the constant sequence
    $\sq{a}{n}$, set $h=1$ and $N$ to the constant function $N(k):=0$ for all
    $k\in\NN$. Then for all $\tau\in\Sigma_2^p$ and $k>2$, the unique
    $k$-decomposition tree of $(a,\tau)$ is the single leaf node whose type is
    in $\mathbf{a}$ by definition of the morphism $\phi$ and the definition of
    $\mathbf{a}$ itself.
  \item \textbf{Product.} Assume that the lemma is true for the two elements
    $\bu,\bv\in\emon$, for $\squ, h_u, N_u$ and $\sqv, h_v, N_v$
    respectively. We will show that it also holds for the element
    $\bu\bv\in\emon$, with the sequence of words $\sq{u_nv_n}{n}$,
    $h=\max\{h_u,h_v\}+1$, and $N(k)=\max\{N_u(k), N_v(k)\}$.

    Let $\taun$ be a sequence of strategies, $k>2$ and $n>N(k)$. Define $\taupn$
    to be the sequence of strategies that are shifted by the lengths of $u_n$,
    i.e. $\tau'_n(i)=\tau_n(i+|u_n|)$ for $n,i\in\NN$. Then by the induction
    hypothesis since $n>N_u(k)$ there exists a $k$-decomposition tree of length
    at most $h_u$ for $(u_n,\tau_n)$ whose root node is labeled by some
    $(U,\tldU)\in\bu$. Similarly there exists a $k$-decomposition tree of length
    at most $h_v$ for $(v_n,\tau'_n)$ whose root node is labeled by some
    $(V,\tldV)\in\bv$. Consequently we can construct a $k$-decomposition tree of
    length at most $\max\{h_u,h_v\}+1$ of $(u_nv_n,\tau)$ whose root node is
    labeled by $(U,\tldU)\cdot (U,\tldV)=(UV,\tldU\tldV)\in \bu\bv$, by making
    the root node a product node and add the two subtrees as children.
  \item \textbf{Iteration.} Assume that the lemma is true for some idempotent
    $\bu\in\emon$. Then there exists a sequence $\squ$, $h_u\in\NN$ and a
    function $N_u$ for which the lemma holds. We will prove that it also holds
    for $\bu^\#\in\emon$, the sequence $\sq{u_n^n}{n}$,
    $h=h_u+3\cdot |\emmon|^2$ and the function $N$ defined by
    $N(k):=N_u(k)+k^{3\cdot|\emmon|^2}$. Let $\taun$ be a sequence of
    strategies, $k>2$ and $n>N(k)$. Since $n>N_u(k)$ by the induction hypothesis
    we know that for all strategies $\tau$,
    $\typ_k^{h_u}\big((u_n,\tau)\big)\cap \bu\neq\emptyset$.

    For $0\leq i < n$ let $\tau_n^i$ be the shifted strategy by $u_n^i$,
    i.e. $\tau_n^i(j)=\tau_n(j+|u_n^i|)$, $j\in\NN$.

    For $0\leq i < n$, pick some
    $(U_i,\tldU_i)\in \typ_k^{h_u}\big((u_n,\tau_n^i)\big)\cap \bu$ and denote
    by $T_i$ the associated $k$-decomposition tree. We modify the alphabet
    $\fra$ and add $(u_n,\tau_n^i)$, $0\leq i< n$ as letters. At the same time
    modify the morphism $\phi$ by mapping $(u_n,\tau_n^i)$ to
    $(U_i,\tldU_i)$.Then applying Theorem~\ref{theo:decomp height} to the word
    $(u_n^n,\tau_n)$ we know that there exists a $k$-decomposition tree of
    height at most $3\cdot |\emmon|^2$, where the leaves are labeled by
    $\big((u_n,\tau_n^i\big),(U_i,\tldU_i))$, $0\leq i< n$. Plugging the trees
    $T_i$ instead of the leaves we construct a $k$-decomposition tree for
    $(u_n^n,\tau_n)$ of height at most $h=h_u+3\cdot |\emmon|^2$. Moreover since
    $n>N(k)\geq k^{3\cdot |\emmon|^2}$ there must exist at least one iteration
    node in this tree therefore the type of the root node can be written as a
    $\#$-expression whose $\#$-height is larger than 1. Consequently the type is
    in $\bu^\#$.
  \end{itemize}
\end{proof}

Observe that the height $h$ gets larger when we iterate, and the same for the
function $N$, but since we do this only a finite number of times (it is
induction on the finite monoid) we can give a uniform height $h$ and function
$N$ such that the lemma above holds for all elements of the extended belief
monoid with that height $h$ and function $N$.

Combining the two lemmata above we have:
\begin{lemma}
  \label{lem:reif}
  For all $\bu\in\emon$ we have a sequence of words $\squ$ over the alphabet
  $\atma$ such that for all $\taun$ there exists $X\subseteq\emmon$ such that
  $\pn=\sq{(u_n,\tau_n)}{n}$ reifies $X$ with height $h$. Moreover
  $X\cap \bu\neq\emptyset$.
\end{lemma}

The {\em raison d'être} of the $k$-decomposition trees, and their bounded height
is because it allows us to give lower and upper bounds on certain outcomes of
the game as in the following lemma. This is where the leaktight hypothesis is
necessary. We start with the lower bound. The proof follows that of \cite{FGO12}.

\begin{lemma}
  \label{lem:lb}
  There exists a function $L:(\NN,\NN)\to\mathbb{R}$ mapping to the non-zero
  positive reals such that for all words $p=(w,\tau)$ over the alphabet $\fra$,
  $k>2$ and $T$ a $k$-decomposition tree of $p$ of height at most $h$ with the
  root node labeled by $(W,\tldW)\in\emmon$, given that $\emmon$ is leaktight
  then for all $s,t\in\stma$,
  \begin{align}
    \label{eq:lb1}
    W(s,t)=1&\implies \PP_s^{w,\tau}(t)\geq L(h,k), \text{ and }\\
    \label{eq:lb2}
    \tldW(s,t)=1&\iff \PP_s^{w,\tau}(t)>0.
  \end{align}
\end{lemma}
\begin{proof}
  We proceed by induction on the structure of the $k$-decomposition tree $T$.
  \begin{itemize}
  \item \textbf{Leaves.} Leaves are labeled by some
    $\big((a,\alpha),(B^{a,\alpha},B^{a,\alpha})\big)$ for $a\in\atma$,
    $\alpha\in\Sigma_2'$. Then \eqref{eq:lb2} is true by definition of
    $B^{a,\alpha}$, and \eqref{eq:lb1} holds for the lower bound $\nu$, where
    $\nu$ is the smallest non-zero transition probability appearing in the
    transition table of the game.
  \item \textbf{Product nodes.} Assume that the lemma holds for the children
    that are labeled by $\big(u,(U,\tldU)\big)$ and $\big(v,(V,\tldV)\big)$ with
    with the lower bound $L$. It follows easily that the lemma holds for the
    parent node that is labeled by $\big(uv,(UV,\tldU\tldV)\big)$ with the lower
    bound $L^2$. The words $u$ and $v$ above are words over the alphabet $\fra$.
  \item \textbf{Idempotent nodes.} Similarly to above if the lemma holds for the
    $n$ children ($n<k$) that are labeled by
    $\big(u_1,(U,\tldU)\big),\big(u_2,(U,\tldU)\big),\ldots,\big(u_n,(U,\tldU)\big)$
    for some idempotent $(U,\tldU)\in\emmon$ with lower bound $L$, it follows
    easily that it also holds for the parent node that is labeled by
    $\big(u_1u_2\cdots u_n,(U,\tldU)\big)$ for the lower bound $L^k$.
  \item \textbf{Iteration nodes.} Assume that the lemma holds for the $n$
    children $(n\geq k)$ that are labeled by
    $\big(u_1,(U,\tldU)\big),\big(u_2,(U,\tldU)\big),\ldots,\big(u_n,(U,\tldU)\big)$
    for some idempotent $(U,\tldU)\in\emmon$ with lower bound $L$. Proving
    \eqref{eq:lb2} is trivial, so we show only \eqref{eq:lb1}. Let $s,t\in\stma$
    such that $U^\#(s,t)=1$, by definition $t$ is $U$-recurrent. Then
    \[
      \PP_s^{u_1\cdots u_n}(t)\geq \PP_s^{u_1}(t)\sum_{q\in\stma}
      \PP_t^{u_2\cdots u_{n-1}}(q)\PP_q^{u_n}(t).
    \]
    Given that $(U,\tldU)$ is leaktight it follows that for all $q\in\stma$,
    $\PP_t^{u_2\cdots u_{n-1}}(q)>0$ implies $\PP_q^{u_n}(t)\geq L$. Indeed, let
    $q\in \stma$ such that $\PP_t^{u_2\cdots u_{n-1}}(q)>0$, then by the
    induction hypothesis we have $\tldU^{n-2}(t,q)=1$ and since $\tldU$ is
    idempotent $\tldU(t,q)=1$. The state $t$ is $U$-recurrent, and $(U,\tldU)$
    is not a leak, therefore it follows from the definition of a leak that
    $U(q,t)=1$. Using the induction hypothesis on the right most child, we have
    $\PP_q^{u_n}(t)\geq L$.  By the induction hypothesis for the left most child
    we have $\PP_s^{u_1}(t)\geq L$. From here we conclude that
    \[
      \PP_s^{u_1\cdots u_n}(t)\geq L^2.
    \]
  \end{itemize}

  By induction we see that the lemma holds for the function $L(h,k)=\nu^{k^h}$.
\end{proof}
We now give a proof for the upper bound. Define $L=L(3\cdot|\emmon|^2,3)$ to be
the lower bound given from Lemma~\ref{lem:lb}.

\begin{lemma}
  \label{lem:ub}
  Let $h\in\NN$, define and $K_h\in\NN$ such that
  $h\cdot (1-L)^{K_h} < L$ and $K_h>|\stma|$.

  For all words $p=(w,\tau)$ over the alphabet $\fra$, $h\in\NN$, $k>K_h$ and
  $T$ a $k$-decomposition tree of $p$ of height at most $h$ with the root node
  labeled by $(W,\tldW)\in\emmon$, given that $\emmon$ is leaktight then for all
  $s,t\in\stma$
  \[
    W(s,t)=0\implies \PP_s^{w,\tau}(t)\leq h\cdot (1-L^{|\stma|})^{\lfloor k/|\stma|\rfloor}.
  \]
\end{lemma}
\begin{proof}
  We proceed by induction on the structure of the $k$-decomposition tree $T$,
  while maintaining upper bounds $F$ that are always smaller than
  $h\cdot (1-L^{|\stma|})^{\lfloor k/|\stma|\rfloor}$.
  \begin{itemize}
  \item \textbf{Leaves.} The leaves are labeled by
    $\big((a,\alpha),(B^{a,\alpha},B^{a,\alpha})\big)$ for $a\in\atma$ and
    $\alpha\in\Sigma_2'$, by definition we have an upper bound of $0$.
  \item \textbf{Product nodes.}  Assume that we have $(u,\big(U,\tldU)\big)$,
    $(v,\big(V,\tldV)\big)$ and the parent node labeled by
    $(uv,\big(UV,\tldU\tldV)\big)$, where $u,v$ are words over the alphabet
    $\fra$. Let $F\geq 0$ be the upper bound of the children,
    i.e. $U(s,t)=0\implies \PP_s^{u}(t)\leq F$. Let $s,t\in\stma$ be such that
    $UV(s,t)=0$. Then the probability of all paths of length two, $s,s',t$ such
    that $\PP_s^u(s')>0$ and $\PP_{s'}^v(t)>0$ is bounded above by $F$,
    therefore $\PP_s^{uv}(t)\leq F$.
  \item \textbf{Idempotent nodes.} Assume that we have the children
    $p_1,\ldots, p_j,j<k$ each decorated by the same idempotent $(W,\tldW)$, and
    let $s,t\in\stma$ such that $W(s,t)=0$. The words $p_i$ are over the
    alphabet $\fra$. By the induction hypothesis the upper bound $F$ holds for
    all the children.

    Denote by $\rho$ the set of all paths $s_0s_1\cdots s_j$ such that $s_0=s$,
    $s_j=t$ and $\tldW(s_i,s_{i+1})=1$for all $0\leq i\leq j-1$. Since
    $W(s,t)=0$ for all $\pi=s_0\cdots s_j\in\rho$ there exists
    $0\leq C(\pi)\leq j-1$ such that $W(s_{C(\pi)},s_{C(\pi)+1})=0$ and for all
    $0\leq i\leq C(\pi)-1$, $W(s_i,s_{i+1})=1$. Define $\rho'$ to be the set of
    such prefixes, i.e.
    \[
      \rho'=\{s_0\cdots s_{C(\pi)}\ |\ \pi=s_0\cdots s_j\in\rho\}.
    \]
    The set $\rho'$ is nonempty because there exists some $r\in\stma$ such that
    $W(s,r)=1$ (this follows from the definition of the half-blind
    game, in every state we have some actions).

    Then we have
    \begin{align*}
      \PP_s^{p_1\cdots p_j}(t)&=\sum_{s_0\cdots
                                s_j\in\rho}\PP^{p_1}_{s_0}(s_1)\cdots \PP_{s_{j-1}}^{p_j}(s_j)\\
                              &\leq \sum_{\pi=s_0\cdots s_{C(\pi)}\in\rho'}\PP_{s_0}^{p_1}\cdots \PP_{s_{C(\pi)-1}}^{p_{C(\pi)}}(s_{C(\pi)})\cdot F\\
                              &\leq F,
    \end{align*}
    where the first inequality is because of the induction hypothesis and
    $W(s_{C(\pi)},s_{C(\pi)+1})=0$, whereas the second inequality is because for
    every path $\pi\in\rho$ there is exactly one path $\pi'\in\rho'$ such that
    $\pi'$ is a prefix of $\pi$.
  \item \textbf{Iteration nodes.} Assume that we have the children
    $p_1,\ldots,p_j$, $j\geq k$ each decorated by the same idempotent
    $(W,\tldW)\in\emmon$ and for whom the upper bound $F$ holds. Let
    $s,t\in\stma$ be such that $W^\#(s,t)=0$. In case $W(s,t)=0$ a proof like
    the one above for idempotent nodes gives $F$ as the upper bound. Therefore
    we assume that $W(s,t)=1$. Then by definition $t$ is $W$-transient and it
    communicates with some recurrence classes whose union we denote by
    $S_r\subseteq \stma$.  We will prove that for all $0\leq i<j'\leq j$ such
    that $j'-i\geq |\stma|$ there exists $i\leq i'\leq j'$ such that
    $i'-i\leq |\stma|$ and 
    \begin{equation}
      \label{eq:implb}
      \PP_t^{p_i\cdots p_{i'}}(S_r)\geq L^{|\stma|}.
    \end{equation}
    
    Let $i \in \{0, \ldots j\}$, then there exists a $3$-decomposition tree
    $T_i$ for the word $p_i$, whose root node is labeled by the element
    $(W_i,\tldW)\in\emmon$. It is possible that $W\neq W_i$, but for all
    $s',t'\in \stma$, $W(s',t')=0$ implies that $W_i(s',t')=0$. This is because
    by the induction hypothesis, if $W(s',t')=0$, we know that
    $\PP_{s'}^{u_i}(t')\leq F$ whereas according to Lemma~\ref{lem:lb} for
    $T_i$, if $W_i(s',t')=1$ we have $\PP_{s'}^{u_i}(t')\geq L$, from
    $F\leq h\cdot (1-L)^k$ and our choice of $k$, superior to $K_h$, this is a
    contradiction, hence $W_i(s',t')=0$. 

    Let $S_t$ be the set of states that are $W$-reachable from $t$ (for all
    $t'\in S_t$, $W(t,t')=1$) but not in $S_r$. These states are all
    $W$-transient and moreover for all $i\in \{0,\ldots, j\}$, there exists a
    $W_i$ path from $t$ and any state in $S_t$ to some element in $S_r$. This is
    because for all $t'\in S_t\cup \{t\}, r\in S_r$, $\tldW(t',r)=1$, and there
    is no $W_i$ path from $S_r$ to $S_t\cup\{t\}$, if there was no $W_i$ path
    from $S_t\cup\{t\}$ we could construct a leak, which contradicts the
    hypothesis that $\emmon$ is leaktight. Similarly, for $0\leq i<j'\leq j$
    such that $i-j'\geq |\stma|$, if $W_i\cdots W_{j'}(t,S_r)=0$ we can
    construct a leak by repeating a factor of $W_i\cdots W_{j'}$, hence we can
    assume that there exists $i\leq i'\leq j'$, such that $i'-i\leq |\stma|$ and
    $W_i\cdots W_{i'}(t,S_r)=1$. Then it follows from Lemma~\ref{lem:lb} that
    $\PP_t^{p_i\cdots p_{i'}}(S_r)\geq L^{|\stma|}$
    which concludes \eqref{eq:implb}.\\

    Let $\rho$ be the set of all paths $s_0\cdots s_j$ such that $s_0=s$,
    $s_j=t$ and $\tldW(s_i,s_{i+1})=1$ for all $0\leq i\leq j-1$. We partition
    $\rho$ into the set $\rho_1$ of all the paths that pass through $S_r$ and
    $\rho_2$ the set of all paths that do not. Since $t$ is $W$-transient, for
    all $r\in S_r$, $W(r,t)=0$, consequently we can use the argument above for
    the idempotent nodes to give $F$ as an upper bound for the probability of
    the event that constitutes the union of all the sets in $\rho_1$. As for
    $\rho_2$, because of transience of $t$ and \eqref{eq:implb} the probability
    of the union of all the paths in $\rho_2$ can be bounded above by
    $(1-L^{\stma})^{\lfloor j/|\stma| \rfloor}$.
  \end{itemize}

  We have shown that the upper bound grows only in the case of iteration nodes
  and it always is smaller than
  $h\cdot (1-L^{|\stma|})^{\lfloor k/|\stma|\rfloor}$, since in ascending the
  tree, at each level we add at most a term of
  $(1-L^{|\stma|})^{\lfloor k/|\stma|\rfloor}$.
\end{proof}

Now Lemma~\ref{lem:assocsub} follows as a corollary of Lemma~\ref{lem:ub} and
Lemma~\ref{lem:reif}. The main point is that the larger the $k$ the smaller the
lower bound we can prove.

%%% Local Variables:
%%% mode: latex
%%% TeX-master: "asymmetric-limitsure"
%%% End:

\subsection{Completeness}
\label{sec:completeness}
% !TEX root = asymmetric-limitsure.tex
Before introducing the main theorem of this section let us give a definition.

\begin{definition}[$\mu$-faithful abstraction]
  \label{def:faithful abstraction}
  Let $u\in\Sigma_1$ be a word, and $\mu>0$ a strictly positive real number. We
  say that $\bu\in\emon$ is a $\mu$-faithful abstraction of the word $u$ if for
  all $(U,\tldU)\in \bu$ there exists $\tau\in\Sigma_2^p$ such that for all
  $s,t\in\stma$,
  \begin{align}
    \label{eq:completeness1}
    \tldU(s,t)=1 &\iff \PP_s^{u,\tau}(t)>0\\
    \label{eq:completeness2}
    U(s,t)=1 &\implies \PP_s^{u,\tau}(t)\geq \mu.
  \end{align}
\end{definition}

This section is devoted to giving the main ideas behind the proof and proving the
following theorem.

\begin{theorem}
  \label{theo:completeness}
  Assume that the game is leaktight. Then there exists $\mu>0$ such that for all
  words $u\in\Sigma_1$ there is some element $\bu\in\emon$ that is a
  $\mu$-faithful abstraction of $u$.
\end{theorem}

The notion of $\mu$-faithful abstraction is compatible with product in the following sense.
%It is an easy exercise to prove that when concatenating two words we can prove that the product of their faithful abstractions is a $\mu^2$-faithful abstraction of the concatenated word i.e.
\begin{lemma}\label{lem:faithproduct}
  Let $\bu,\bv\in\emon$ be $\mu$-faithful abstractions of $u\in\Sigma_1$ and
  $v\in\Sigma_1$ respectively. Then $\bu\bv$ is a $\mu^2$-faithful abstraction
  of $uv\in\Sigma_1$.
\end{lemma}

A na\"ive use of Lemma~\ref{lem:faithproduct} shows that any word $w$
has a $\mu_w$-faithful abstraction in $\emon$,
where $\mu_w$ converges to $0$ as the length of $w$ increases.
However we need $\mu_w$ to depend only on $\emon$,
independently of $|w|$. For that we make use of $k$-decomposition trees.
More precisely we build
$N$-decomposition trees for words in $\Sigma_1$ where $N=2^{3\cdot |\emmon|}$. We can
construct $N$-decomposition trees for any word $u\in\Sigma_1$ whose height is at
most $3\cdot |\emon|^2$ and since $N$ is fixed we will be able to propagate the
constant $\mu$, it only remains to take care that the constant does not shrink
as a function of the number of children in iteration nodes, hence the following
lemma.

\begin{lemma}
  \label{lem:iter completness} 
  Let $u\in\Sigma_1$ be a word factorized as $u=u_1\cdots u_n$ where
  $n>2^{3\cdot |\emmon|}=N$, and $\bu\in\emon$ an
  idempotent element such that $\bu$ is a $\mu$-faithful
  abstraction of $u_i$, $1\leq i\leq n$, for some $\mu>0$. If
  $\bu$ is not a leak then $\bu^\#$ is a $\mu'$-faithful
  abstraction of $u$, where $\mu'=\mu^{N+1}$.
\end{lemma}

Theorem~\ref{theo:completeness} is an easy consequence from the lemmata above, which can be shown as follows. We
construct a $N$-decomposition tree for the word $u\in\Sigma_1$, and propagate the lower bound from the leaf nodes, for
which we have the bound $\nu>0$ (where $\nu$ is the smallest transition probability appearing in the game) up to the
root node. If we know that a bound $\mu>0$ holds for the children, for the parents we have the following lower bounds as
a function of the kind of the node: (1) {\em product node}: $\mu^2$; (2) {\em idempotent node} $\mu^N$; (3) {\em
  iteration node} $\mu^{N+1}$. Since the length of the tree is at most $h=3\cdot |\emon|^2$ we have the lower
bound \[\mu=\nu^{h(N+1)}\] that holds for all $u\in\Sigma_1$.

\begin{proof}[Proof of Lemma~\ref{lem:iter completness}]
  Let $(W,\widetilde W)\in\bu^\#$ we want to build a strategy
  $\tau\in\Sigma_2^p$ such that \eqref{eq:completeness1} and
  \eqref{eq:completeness2} in Definition~\ref{def:faithful abstraction} hold,
  for $(W,\tldW)$ the word $u$ and the bound $\mu'$. Let us first assume that
  $(W,\tldW)$ is such that
  \begin{align*}
    &W=F_1G_1^\#\cdots F_kG_k^\#F_{k+1},\\
    &\tldW=\widetilde F_1\widetilde G_1\cdots \widetilde F_k\widetilde G_k\widetilde F_{k+1},
  \end{align*}
  where $(F_i,\widetilde F_i)\in\bu$,$(G_i,\widetilde G_i)\in\bu$ and
  $(G_i,\widetilde G_i)$ are idempotent. 

The set of $\#$-expressions of $\bu\subseteq\mmon$ denoted by $\expr(\bu)$ is a
language defined by the grammar:
$\expr(\bu):=\bu\ \mid\ \expr(\bu)\cdot\expr(\bu)\ \mid\ (\expr(\bu))^\#$, so
the terminal symbols are the elements of $\bu$. There is $\gamma_\bu$, a natural
function mapping $\expr(\bu)$ to $\mmon$, i.e. the function that is the identity
when restricted to the terminal symbols, otherwise
$\gamma_\bu(e\cdot e')=\gamma_\bu(e)\gamma_\bu(e')$, and
$\gamma_\bu(e^\#)=(\gamma_\bu(e))^\#$. Given $e\in\expr(\bu)$ we define its {\em
  $\#$-height} as the number of the deepest nesting of
$\#$. E.g. $\height(U^\#(VW^\#)^\#)=2$.

  We can safely make this assumption because for all $(U,\tldU)\in\bu$ we can
  find a $\#$-expression $e\in\expr(\bu)$ whose $\#$-height is $1$, such that
  $\gamma_\bu(e)=(U',\tldU)$ and for all $s,t\in\stma$
  $U(s,t)=1\implies U'(s,t)=1$. This is an easy exercise: when iterating we are
  removing edges.

  Since $\bu$ is a $\mu$-faithful abstraction of $u_i$,$1\leq i\leq n$, for all
  $(U,\tldU)$ in $\mathbf{u}$ there is a strategy in $\Sigma_2^p$ such that
  \eqref{eq:completeness1} and \eqref{eq:completeness2} hold. Let $\tau_1$ be
  such a strategy for $(F_1,\widetilde F_1)$, $\tau_2$ for
  $(G_1,\widetilde G_1)$ and so on until $\tau_{2k+1}$ for the selection
  $(F_{k+1},\widetilde F_{k+1})$. We define the strategy $\tau$ by assigning
  one of the $\tau_i$ to some part of the word in the following way:
  \begin{itemize}
  \item against $u_1$ play $\tau_1$,
  \item against $u_2$ play $\tau_2$, play $\tau_2$ also against $u_3,u_4,\ldots,u_{n-2k+1}$ each,
  \item against $u_{n-2k+2}$ play $\tau_3$, etc., in general against
    $u_{n-2k+1+i}$ play $\tau_{i+2}$, $1\leq i\leq 2k-1$.
  \end{itemize}
  One can visualize this in the following way.
  \[
\tau:=    {u_1\atop{\tau_1\atop{F_1}}} {\mid\atop{\mid\atop{\mid}}} {u_2,u_3,\ldots,u_{n-2k+1}\atop{\tau_2\atop{G_1}}} {\mid\atop{\mid\atop{\mid}}} {u_{n-2k+2}\atop{\tau_3\atop{F_2}}} {\mid\atop{\mid\atop{\mid}}} \cdots {\mid\atop{\mid\atop{\mid}}} {u_{n-2k+2k-1}\atop{\tau_{2k}\atop{G_k}}} {\mid\atop{\mid\atop{\mid}}} {u_n\atop{\tau_{2k+1}\atop{F_{k+1}}}}.
  \]
  This means that $\tau$ plays according to $\tau_2$ against $u_2$ then it keeps
  playing according to $\tau_2$ against $u_3$ and so on until $u_{n-2k+1}$ is
  read. Note that it is well defined since we have assumed that $n>N$, and
  $N=3\cdot |\emmon|^2\geq 2k+1$ from Simon's forest factorization theorem.\\

  Now we prove \eqref{eq:completeness1} for $\tau$ and $u$.
  
  ($\implies$) Let $s,t\in\stma$ be such that $s\xrightarrow{\tldW}t$.
  Since
  $\tldW=\widetilde F_1\widetilde G_1\cdots \widetilde G_k\widetilde
  F_{k+1}$
  and $\widetilde G_1$ is idempotent there exist $s_1,\ldots,s_{n-1}\in\stma$
  such that
  \begin{equation}
    \label{eq:plusarrows}
    s\xrightarrow{\widetilde F_1}s_1\xrightarrow{\widetilde G_1}\cdots
    \xrightarrow{\widetilde G_1} s_{n-2k}\xrightarrow{\widetilde F_2}
    s_{n-2k+1}\xrightarrow{\widetilde G_2}s_{n-2k+2}\cdots
    \xrightarrow{\widetilde G_{k}}s_{n-1}\xrightarrow{\widetilde F_{k+1}}t.
  \end{equation}
  Let $F(s_1,\ldots,s_{n-1})$ be equal to
  \[
    \PP_s^{u_1,\tau_1}(s_1)\PP_{s_1}^{u_2,\tau_2}(s_2)\cdots
    \PP_{s_{n-2k-1}}^{u_{n-2k+1},\tau_2}(s_{n-2k})\PP_{s_{n-2k}}^{u_{n-2k+2},\tau_3}(s_{n-2k+1})\cdots\PP_{s_{n-1}}^{u_n,\tau_{2k+1}}(t).
  \]
  Then by the choice of $\tau$ we have
  $\PP_s^{u,\tau}(t)\geq F(s_1,\ldots,s_{n-1})$.  Since $\mathbf{u}$
  is a $\mu$-faithful abstraction of $u_i$, \eqref{eq:plusarrows}
  implies that every factor of $F(s_1,\ldots,s_{n-1})$ is positive,
  hence $\PP_s^{u,\tau}(t)>0$.

  ($\impliedby$) Let $s,t\in\stma$ be such that
  $\PP_s^{u,\tau}(t)>0$, then similarly as above there must exist states $s_1,\ldots,s_{n-1}$ such that

\[
\PP_s^{u,\tau}(t)\geq F(s_1,\ldots,s_{n-1})>0.
\]
This implies \eqref{eq:plusarrows} since $\bu$ is a
$\mu$-faithful abstraction of all $u_i$, and in turn,
\eqref{eq:plusarrows} implies that $s\xrightarrow{\tldW} t$,
since
$\tldW=\widetilde  F_1\widetilde  G_1\cdots \widetilde 
G_k\widetilde  F_{k+1}$. \\

Now we prove \eqref{eq:completeness2} for $\tau$ and $u$ and the
  bound $\mu'=\mu^{N+1}$. Let
$s,t\in\stma$ such that $s\xrightarrow{W} t$. Then there exists states $s_1,\ldots,s_{2k}$ such that 
\begin{equation}
\label{eq:onearrows}
s\xrightarrow{F_1} s_1 \xrightarrow{G^\#_1} s_2 \xrightarrow{F_2} \cdots \xrightarrow{G^\#_k} s_{2k} \xrightarrow{F_{k+1}} t.
\end{equation}
First we will show that
\begin{equation}
\label{eq:geqmu2}
\PP_{s_1}^{u_2,\ldots,u_{n-2k+1},\tau'}(s_2)\geq \mu^2,
\end{equation}
where $\tau'$ is the strategy that plays $\tau_2$ against $u_2$, and against
$u_3$ and so on. This is exactly what the strategy $\tau$ does, after $u_1$ is
read. Then we have
\[
  \PP_{s_1}^{u_2\cdots u_{n-2k+1},\tau'}(s_2)\geq \PP_{s_1}^{u_2,\tau_2}(s_2)\sum_{s'\in\stma}\PP_{s_2}^{u_3\cdots u_{n-2k},\tau''}(s')\PP_{s'}^{u_{n-2k+1},\tau_2}(s_2),
\]
where $\tau''$ is the strategy that plays $\tau_2$ against $u_3$ ,and against
$u_4$ and so on. The strategy $\tau''$ is the same as $\tau'$ just shifted by
the first part $u_2$. From \eqref{eq:onearrows} $s_1\xrightarrow{G^\#_1} s_2$
which implies that $s_2$ is $G_1$-recurrent,
$s_1\xrightarrow{\widetilde G_1} s_2$ and $s_1\xrightarrow{G_1} s_2$. By the
choice of $\tau_2$ because $s_1\xrightarrow{G_1} s_2$ we have
\begin{equation}
  \label{eq:leakweak}
  \PP_{s_1}^{u_2\cdots u_{n-2k+1},\tau'}(s_2)\geq \mu\sum_{s'\in\stma}\PP_{s_2}^{u_3\cdots u_{n-2k},\tau''}(s')\PP_{s'}^{u_{n-2k+1},\tau_2}(s_2).
\end{equation}
Let $s'$ be such that $\PP_{s_2}^{u_3\cdots u_{n-2k},\tau''}(s')>0$.
Then from the definition of $\tau''$,
$s_2\xrightarrow{\widetilde G_1^{n-2k-3}}s'$ and since
$\widetilde G_1$ is idempotent $s_2\xrightarrow{\widetilde G_1}s'$. We
will prove that $s'\xrightarrow{G_1} s_2$. There are two cases:
\begin{itemize}
\item $s'$ is $G_1$-recurrent: then both $s'$ and $s_2$ are
  $G_1$-recurrent, and $s_2\xrightarrow{\widetilde G_1}s'$. Since we
  have assumed that $\mathbf{u}$ is not a leak, then
  $s'\xrightarrow{G_1}s_2$.
\item $s'$ is $G_1$-transient: There exists some state $r$ that is
  $G_1$-recurrent, such that $s'\xrightarrow{G_1} r$ and
  $r\not\xrightarrow{G_1}s'$. Now $s'\xrightarrow{G_1} r$ implies that
  $s'\xrightarrow{\widetilde G_1} r$, and from idempotency of
  $\widetilde G_1$, $s_2\xrightarrow{\widetilde G_1} r$. Then from the
  argument for the case above $r\xrightarrow{G_1}s_2$, and finally
  from idempotency of $G_1$, $s'\xrightarrow{G_1}s_2$.
\end{itemize}

We have shown that for all $s'$ such that
$\PP_{s_2}^{u_3\cdots u_{n-2k},\tau''}(s')>0$,
$s'\xrightarrow{G_1}s_2$. As a consequence, from the choice of $\tau_2$ and \eqref{eq:leakweak} we have
\[
  \PP_{s_1}^{u_2\cdots u_{n-2k+1},\tau'}(s_2)\geq \mu^2.
\]

To finish up with the proof of \eqref{eq:completeness2}, for all
$s,s'\in\stma$ and $G_i$, $s\xrightarrow{G^\#_i}s'$ implies that
$s\xrightarrow{G_i}s'$, therefore from \eqref{eq:onearrows} we have
\begin{equation}
\label{eq:onearrows2}
s\xrightarrow{F_1} s_1 \xrightarrow{G^\#_1} s_2 \xrightarrow{F_2}s_3\xrightarrow{G_2} \cdots \xrightarrow{G_k} s_{2k} \xrightarrow{F_{k+1}} t,
\end{equation}
so for all $G_i$, $2\leq i\leq k$, we write $G_i$ instead of $G_i^\#$.
Then by the choice of the strategies $\tau_i$ and the definition of
$\tau$,
\[
  \PP_s^{u,\tau}(t)\geq \PP_s^{u_1,\tau_1}(s_1)\PP_{s_1}^{u_2,\ldots,u_{n-2k+1},\tau'}(s_2)\cdots \PP_{s_{2k}}^{u_n,\tau_{2k+1}}(t)\geq \mu\cdot \mu^2 \cdot \mu^{2k-1}=\mu^{2k+2},
\]
where for the last inequality we have used \eqref{eq:geqmu2} and
\eqref{eq:onearrows2}. Since $2k+1\leq N$, this concludes
the proof of \eqref{eq:completeness2} for $\tau$,$u$ and the bound
$\mu'=\mu^{N+1}$.
\end{proof}

%%% Local Variables:
%%% mode: latex
%%% TeX-master: "asymmetric-limitsure"
%%% End:

\section{Complexity of Optimal Strategies}
\label{sec:strategies}
% !TEX root = asymmetric-limitsure.tex
The maxmin reachability problem solved by the belief
monoid algorithm concerns games where the maximizer is restricted to pure
strategies, and decides whether
\[
\underline{\val}(s)=\sup_{w\in
  \Sigma_1}\inf_{\tau\in\Sigma_2}\PP_s^{w,\tau}(F)=1
  \]
  % is equal to $1$, 
  where
$\Sigma_1=\atma^*$. If we extend further the set $\Sigma_1$ of strategies of the
maximizer and allow him to have mixed strategies too, then half-blind games have
a value. Let $\Sigma_1^m=\Delta(\atma^*)$ be the set of mixed words.
\begin{theorem}[\cite{GR14}] Half-blind games where maximizer can use mixed strategies have a value:
\[
  \val(s)=\sup_{w\in\Sigma_1^m}\inf_{\tau\in\Sigma_2}\PP_s^{w,\tau}(F)=\inf_{\tau\in\Sigma_2}\sup_{w\in\Sigma_1^m}\PP_s^{w,\tau}(F).
\]
\end{theorem}

Define $\Sigma_2^f$ to be the set of finite-memory strategies for the
minimizer. These are strategies that are stochastic finite-state probabilistic transducers
reading histories and outputting elements of $\Delta(\atmi)$, mixed actions. Let
$\val^f(s)=\inf_{\tau\in\Sigma_2^f}\sup_{w\in
  \Sigma_1}\PP_s^{w,\tau}(F)$.

In general,
\[
  \underline\val(s)\leq \val(s)\leq \val^f(s).
\]

A natural question is whether the inequalities above are strict in general,
i.e. whether mixed strategies are strictly more powerful for the
maximizer and whether infinite-memory strategies are strictly more powerful for the
minimizer. 

The former question can be resolved easily. We can find examples
where the maximizer wins more by mixing her strategy. In fact the example in
Figure~\ref{fig:running_example} suffices. For this example we have
$\underline\val(s)< \val(s)$.

The latter question --- whether there exists an example such that
$\val(s)<\val^f(s)$ --- is harder, and its answer is more
counter-intuitive. When maximizer has full information,
it is well-known that minimizer can play optimally with no memory
(using a positional strategy).
When maximizer is totally blind, one might believe that minimizer does not need any memory either
because playing against an opponent that is totally blind
to satisfy a safety objective seems rather easy.
Surprisingly perhaps,
minimizer requires infinite memory
to play optimally against a blind maximizer and satisfy its safety objective.
%This may be surprising because 
We show that there exists a game where
$\val(s)<\val^f(s)$. This game is based on the following gadget.

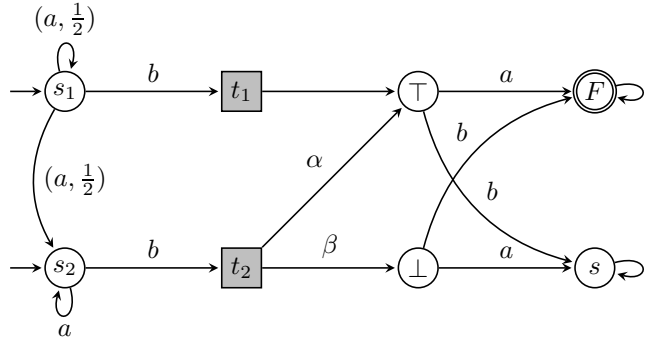
\begin{figure}[H]
  \begin{center}
    \begin{tikzpicture}[->,>=stealth,shorten >=1pt,auto,node distance=1.8cm,
      semithick, initial text=]
      \tikzstyle{vertex1}=[circle,draw=black,minimum size=15pt,inner sep=0pt] 
      \tikzstyle{vertex2}=[fill=black!25,draw=black,minimum size=15pt,inner sep=0pt]
      \node[vertex1,initial] (s1) {$s_1$};
      \node[vertex1,initial] (s2)[below=of s1] {$s_2$};
      \node[vertex2] (t1)[right=of s1]{$t_1$}; 
      \node[vertex2] (t2)[right=of s2]{$t_2$};
      \node[vertex1] (tp)[right=of t1]{$\top$};
      \node[vertex1] (bt)[right=of t2]{$\bot$};
      \node[vertex1,double] (fin)[right=of tp]{$F$};
      \node[vertex1] (sink)[right=of bt]{$s$};
      \path 
      (s1) edge[loop above] node {$(a,\frac{1}{2})$} (s1)
      (s1) edge[bend right] node {$(a,\frac{1}{2})$} (s2)
      (s1) edge node {$b$} (t1)
      (s2) edge node {$b$} (t2)
      (s2) edge[loop below] node {$a$} ()
      (t1) edge node {} (tp)
      (t2) edge node {$\alpha$} (tp)
      (t2) edge node {$\beta$} (bt)
      (tp) edge node {$a$} (fin)
      (tp) edge[bend right] node {$b$} (sink)
      (bt) edge node {$a$} (sink)
      (bt) edge[bend left] node {$b$} (fin)
      (fin) edge[loop right] node {} ()
      (sink) edge[loop right] node {} ();
    \end{tikzpicture}  
  \end{center}
  \caption{A gadget}
  \label{example:gadget}
\end{figure}

We give the main idea behind the gadget. The maximizer wants to be able to
ascertain whether she is in state top or bottom after playing his first $b$ so
that she can go to the final state. The objective of the minimizer is to make
the probability of being in the top state equal to that of being in the bottom
state, so that the maximizer cannot win more than $1/2$. In order to do this,
when it is his turn to make the choice between $\alpha$ and $\beta$ (or a mixing
of them) she has to know the exact probability distribution over $t_1$ and
$t_2$. But this is impossible to keep track with a finite-memory strategy,
i.e. the maximizer plays too many $a$'s for the minimizer's small memory. Hence
the maximizer can always win slightly more than $1/2$. We then use this gadget
in a game that emphasizes the importance of these winnings and prove that in
that game $1/2=\val(s)<\val^f(s)=1$. We prove this formally.

The game starts either at state $s_1$ or $s_2$ with equal probability. The maximizer can play a series of $a$'s and
eventually has to play a $b$ if she wants to make progress. After which the minimizer observes whether the game is in
the state $t_1$ or $t_2$.  In case it is in $t_1$ the minimizer has no choice and proceeds to state $\top$. In case it
is in $t_2$ the minimizer can choose between $\alpha$ and $\beta$ to go either to state $\top$ or to state $\bot$. Then
the maximizer has to guess which one it is. If the guess is right she wins if it is wrong she loses by going to the sink
state. The goal of the minimizer is to keep track of the probability distribution on the states of the game such that
when it is her time to make a decision she will play a mixed (between $\alpha$ and $\beta$) action such that the
probability to be in $\top$ is equal to the probability to be in $\bot$ equal to $1/2$. Keeping track of the
distribution will be impossible with a finite-memory strategy because the sequence of $a$'s that the maximizer plays can
be arbitrarily long.

Observe that $\val(\gamma)=1/2$, where $\gamma$ is the initial
distribution, i.e. $\gamma(s_1)=\gamma(s_2)=1/2$, by giving the optimal
strategies as follows. The maximizer can mix the two words $ba$ and $bb$ with
equal probability. Call this mixed word $w$. Then for all strategies $\tau$ that
the minimizer chooses we have $\PP_{\gamma}^{w,\tau}(F)=1/2$. On the other hand,
after a $b$ is played, the probability to be in the state $t_2$ is always larger
than $1/2$, $\PP_{\gamma}^{a^nb,\tau}(t_2)\geq 1/2$, and consequently the
minimizer has an optimal action such that both $\top$ and $\bot$ are reached
with equal probability and equal to $1/2$. Moreover this optimal action can be
played by the minimizer by keeping track of the distribution on $t_1$ and $t_2$
by counting the number of $a$s that are played before $b$. Albeit this requires
unbounded memory. We give a proof of this in what follows.

Assume that the game stops just before the minimizer makes her action, then we
have
\[
  \PP_{\gamma}^{a^nb,\tau}(t_2)=1-\frac{1}{2}\cdot \frac{1}{2^n}=\frac{2^{n+1}-1}{2^{n+1}}
\]
Therefore if $\tau$ is optimal, after seeing $a^nb$ it would play the
action $\beta$ with the following probability, 
\[
  \tau(a^nb)(\beta)=\frac{1}{2}\cdot \frac{2^{n+1}}{2^{n+1}-1}=\frac{1}{2}\cdot\frac{1}{1-\frac{1}{2^{n+1}}}.
\]
With such a strategy it would ensure that
$\PP_{\gamma}^{a^nb,\tau}(\top)=\PP_{\gamma}^{a^nb,\tau}(\bot)=1/2$. We
prove that this is impossible with a finite-memory strategy.

The proof is by contradiction. Assume that the minimizer has a finite-memory
strategy with $m$ states such that against the word $a^nb$ it plays the action
$\beta$ with probability $\frac{1}{2}\cdot\frac{1}{1-\frac{1}{2^{n+1}}}$. From
the definition of a finite-memory strategy, this implies that there exist two
$m\times m$ stochastic matrices $A$ and $B$, and $J\subset \{1,\ldots ,m\}$ such
that

\begin{equation}
  \label{eq:compstr3}
  \sum_{j\in J} (A^nB)_{i,j}=\frac{1}{2}\cdot\frac{1}{1-\frac{1}{2^{n+1}}},
\end{equation}

where $i$ is the initial memory location of the strategy, and for a
matrix $A$ we denote by $A_{i,j}$ the element on the $i$th row and
$j$th column. We use the following well-known theorem.  See
e.g. \cite{gantmacher1959theory}.
\begin{theorem}
  \label{theo:matpow}
  Let $A$ be a square $m\times m$ stochastic matrix and
  $\lambda_1,\lambda_1,\ldots,\lambda_r$ ($r\leq m$) its distinct
  eigenvalues. Then for all $n>m$
  \[
    (A^n)_{i,j}=\sum_{k=1}^r\lambda_k^n P_{ijk}(n),
  \]
  where $P_{ijk}$ are polynomials of smaller order than the
  multiplicity of $\lambda_k$.
\end{theorem}

Using Theorem \ref{theo:matpow} and doing a small calculation we see
that indeed there exist polynomials $P_1,P_2,\ldots,P_r$ such that for
all $n>m$
\[
  \sum_{j\in J} (A^nB)_{i,j}=\sum_{k=1}^r\lambda_k^n P_k(n).
\]
On the other hand the Taylor expansion for $\frac{1}{1-\frac{1}{2^{n+1}}}$ give
us
\[
  \sum_{j\in J} (A^nB)_{i,j}=\frac{1}{2}\cdot (1+\frac{1}{2^{n+1}}+\frac{1}{2^{2(n+1)}}+\cdots).
\]
Therefore 

\begin{equation}
  \label{eq:compstr1}
  \sum_{k=1}^r\lambda_k^n P_k(n)=\frac{1}{2}\cdot (1+\frac{1}{2^{n+1}}+\frac{1}{2^{2(n+1)}}+\cdots).
\end{equation}

Now observe that for complex numbers $z_1,\ldots,z_m$, $m\geq 1$, with
$|z_1|=|z_2|=\cdots =|z_m|$, real $c>0$, and polynomials $f_1,\ldots,f_m$ on $n$
of degree at most $d$,
\begin{equation}
  \label{eq:tempgen}
  \lim_{n\to\infty}\frac{c}{\sum_{i=1}^mz_i^nf_i(n)}=1,
\end{equation}
implies that $\sum_{i=1}^mz_i^nf_i(n)=\sum_{i=1}^mc_iz_i^n=c$ for some constants
$c_i$. The reason being that \eqref{eq:tempgen} clearly cannot be true for
$|z_i|<1$, as for $|z_i|\geq 1$ assume that the dominating term of the
denominator has the form $n^k\sum_{i=1}^mc_iz_i^n$ for constants $c_i$, then for
\eqref{eq:tempgen} to hold we need $k=0$. Hence
$\sum_{i=1}^mz_i^nf_i(n)=\sum_{i=1}^mc_iz_i^n$, and similarly it is necessary
that $|z_1|=|z_2|=\cdots = |z_m|=1$. Finally because of \eqref{eq:tempgen} we
have $\sum_{i=1}^mc_iz_i^n=c$.

Assume without loss of generality that
$|\lambda_1|=|\lambda_2|=\cdots = |\lambda_{r_1}|$, for some $1\leq r_1\leq r$
and that $|\lambda_1|\geq |\lambda_i|$, $1\leq i\leq r$. The expression on the
left hand side of \eqref{eq:compstr2} is dominated by
$\sum_{k=1}^{r_1}\lambda_k^nP_k(n)$ whereas the expression on the right hand
side is dominated by the leading term $1/2$.

Consequently, because of the equality above, it holds that
\[
  \lim_{n\to\infty}\frac{\frac{1}{2}}{\sum_{k=1}^{r_1}\lambda_k^n P_k(n)}=1.
\]
Applying \eqref{eq:tempgen} we have
$\sum_{k=1}^{r_1}\lambda_k^n
P_k(n)=\sum_{k=1}^{r_1}\lambda_k^nc_k=\frac{1}{2}$.
We substract both of these equal quantities from \eqref{eq:compstr1},
to get
\begin{equation}
  \label{eq:compstr2}
  \sum_{k=r_1}^r\lambda_k^n P_k(n)=\frac{1}{2}\cdot (\frac{1}{2^{n+1}}+\frac{1}{2^{2(n+1)}}+\cdots).
\end{equation}

Repeating the same argument for the leading terms of
\eqref{eq:compstr2} we have
\[ \lim_{n\to\infty}\frac{\frac{1}{2^{n+2}}}{\sum_{k=r_1}^{r_2}\lambda_k^nP_k(n)}=\lim_{n\to\infty}\frac{1}{4\sum_{k=r_1}^{r_2}(2\lambda_k)^nP_k(n)}=1.
\]
Again, applying \eqref{eq:tempgen} we get
$\sum_{k=r_1}^{r_2}(2\lambda_k)^nP_k(n)=\sum_{k=r_1}^{r_2}c'_k2^n\lambda_k^n=1/4$. Hence
we can subtract the quantity $\frac{1}{2^{n+2}}$ from both sides in
\eqref{eq:compstr2}. Repeating the same argument for the eigenvalues
that are left we conclude that
\[
0=\frac{1}{2}\cdot(\frac{1}{2^{r(n+1)}}+\frac{1}{2^{(r+1)(n+1)}}+\cdots),
\]
which is clearly a contradiction therefore there are no two finite stochastic
matrices $A,B$ such that \eqref{eq:compstr3} holds, and consequently the
minimizer has no finite-strategy that is optimal in achieving the $1/2$
payoff. Nevertheless for all $\epsilon>0$ the minimizer has $\epsilon$-optimal
strategies that have finite-memory. These strategies would constitute of
counting the number of $a$'s up to some length.

We have shown the following lemma. 

\begin{lemma}
  \label{lem:gadget}
  In the game in Figure \ref{example:gadget} for all finite-memory strategies
  $\tau$ for the minimizer there exists a word $w$ such that
  \[
    \PP_\gamma^{w,\tau}(F)>\frac{1}{2},
  \]
  where $\gamma$ is the initial distribution,
  $\gamma(s_1)=\gamma(s_2)=1/2$. 
\end{lemma}

Now we give an example that gives a stronger property. We will use the game in
Figure \ref{example:gadget} in another game as a gadget. We then demonstrate
that for this larger game it also holds that $\val(i)=1/2$ where $i$ is the
initial state but $\val^f(i)=1$, i.e. for all finite-memory strategies $\tau$
and $\epsilon>0$ there is a finite word that reaches the set of final states
with probability larger than $1-\epsilon$.

\begin{figure}[H]
  \begin{center}
    \begin{tikzpicture}[->,>=stealth,shorten >=1pt,auto,node distance=1.8cm,
      semithick, initial text=]
      \tikzstyle{vertex1}=[circle,draw=black,minimum size=15pt,inner sep=0pt] 
      \tikzstyle{vertex2}=[star, star points=7,fill=black!25,draw=black,minimum size=15pt,inner sep=0pt]
      \node[vertex1] (i) {$i$};
      \node[vertex2] (t)[right of=i] {$\top$};
      \node[vertex2] (b)[left of=i] {$\bot$};
      \node[vertex1] (tt)[right=of t] {$\top\top$};
      \node[vertex1] (tb)[below of=tt] {$\top\bot$};
      \node[vertex1] (bb)[left=of b] {$\bot\bot$};
      \node[vertex1] (bt)[below of=bb]{$\bot\top$};
      \node[vertex1] (sink)[above of=bb]{$s$};
      \node[vertex1,double] (fin)[above of=tt]{$f$};
      \node[vertex1] (m)[below of=i] {$m$};
      \path 
      (i) edge[bend right,in=270,pos=0.3,above] node {$(a,\frac{1}{2})$} (b)
      (i) edge[bend left,in=90,above,pos=0.3] node {$(a,\frac{1}{2})$} (t)
      (b) edge[bend left] node {$c_2,y_2$} (bb)
      (b) edge[bend left] node {$c_2,x_2$} (bt)
      (b) edge[loop right] node {$c_1$} ()
      (t) edge[bend right] node {$c_1,x_1$} (tt)
      (t) edge[bend right] node {$c_1,y_1$} (tb)
      (tt) edge[loop right] node {$c_2$} ()
      (tb) edge[loop right] node {$c_2$} ()
      (tt) edge[bend right,above] node {$R$} (t)
      (tb) edge[bend left] node {$R$} (m)
      (bt) edge[bend right, below] node {$R$} (m)
      (bb) edge[bend left] node {$R$} (b)
      (b) edge[bend right,above] node {$\bar R$} (sink)
      (t) edge[bend left] node {$\bar R$} (fin)
      (m) edge node {$\bar R$} (i)
      (m) edge[loop below] node {$R,c_1,c_2$} ()
      (sink) edge[loop left] node {} ()
      (fin) edge[loop right] node {} ()
      ;
    \end{tikzpicture}  
  \end{center}
  \caption{A game for which $\val^\infty(i)=1/2$}
  \label{example:cexample}
\end{figure}
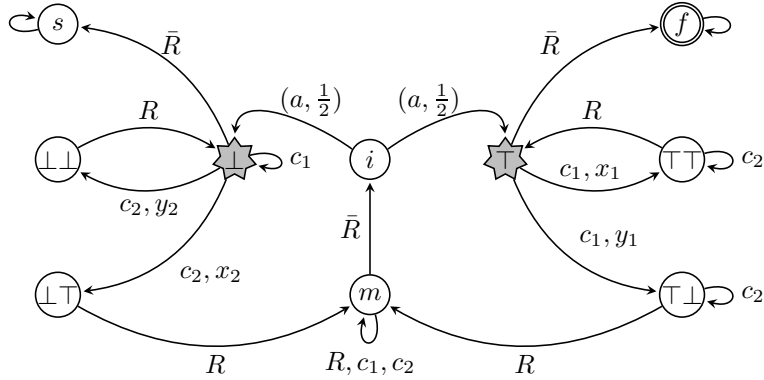

We give an informal description of the game in Fig. \ref{example:cexample}. The
state $i$ is the initial state. A fair coin is tossed at $i$ and if it is heads
then we move to state $\top$ otherwise we move to state $\bot$. Then, we toss a
biased coin in $\top$ by playing $c_1$, if we happen to be in $\bot$, playing
$c_1$ would not change anything. At this point another biased coin is tossed by
playing $c_2$ as a result we are in one of the states
$\bot\bot,\bot\top,\top\top,\top\bot$ after the two coin tosses. Repeating this
process $n$ times, i.e. by playing $a(c_1c_2R)^n$, we end up in state $\top$ if
and only if we had $n+1$ heads and symmetrically we end up in state $\bot$ if
and only if we have tossed $n+1$ tails. Now we play $\bar R$, and by doing so we
win if we have tossed $n+1$ consecutive heads, we lose if we have tossed $n+1$
consecutive tails and otherwise we go to the state $i$. If we repeat this
process $k$ times, i.e. by playing the word
\[
  (a(c_1c_2R)^n\bar R)^k,
\]
then the probability to win the game will be arbitrarily close to $1$ (for well
chosen $n$ and $k$) if and only if the coin tosses are biased towards heads,
i.e. $x_1,x_2>1/2$. Then the idea is to embed the gadget in
Fig.\ref{example:gadget} in place of the states $\bot$ and $\top$.

For all $k$ let 
\[
\mu_k=\PP_i^{(a(c_1c_2R)^n\bar R)^k}(\neg \{f,s\}),
\]
the probability to be in any state except the sink $(s)$ or final
$(f)$ state after the word $(a(c_1c_2R)^n\bar R)^k$ has been
played. Then we have
\begin{align*}
  \mu_0 &= 1, \text{ and }\\
  \mu_k &= \mu_{k-1}(1-\frac{1}{2}x_1^n-\frac{1}{2}y_2^n). 
\end{align*}
Hence
\[
  \mu_k = (1-\frac{1}{2}x_1^n-\frac{1}{2}y_2^n)^k. 
\]

Observe that
\begin{align*}
  \PP_i^{(a(c_1c_2R)^n\bar R)^k}(f)&=\frac{1}{2}x_1^n(\mu_0+\mu_1+\cdots +\mu_{k-1})\\
  &=\frac{1}{2}x_1^n\frac{1-(1-\frac{1}{2}x_1^n-\frac{1}{2}y_2^n)^k}{1-(1-\frac{1}{2}x_1^n-\frac{1}{2}y_2^n)}\\
  &=\frac{x_1^n}{x_1^n+y_2^n}\cdot (1-(1-\frac{1}{2}x_1^n-\frac{1}{2}y_2^n)^k).
\end{align*}
Then there exists some function $g$ such that
$\lim_{n\to\infty}(1-\frac{1}{2}x_1^n-\frac{1}{2}y_2^n)^{g(n)}=0$. Also, we have
$x_1>y_2$ if and only if $\lim_{n\to\infty}\frac{x_1^n}{x_1^n+y_2^n}=1$.

If we embed the gadget in Fig. \ref{example:gadget} in place of the
states $\bot$ and $\top$ and replace the letter $c_1$ with the letters
$a_1,b_1$ from the gadget and symmetrically $c_2$ with the letters
$a_2,b_2$, and such that the final state of the gadget embedded on the
right becomes $\top\top$, the sink state $\top\bot$ and symmetrically
the final state of the gadget embedded on the left becomes $\bot\top$
and the sink state $\bot\bot$ together with Lemma \ref{lem:gadget}
implies the following:

\begin{theorem}
  There exists a game with initial state $i$, such that $1/2=\val(i)<\val^f(i)=1$.
\end{theorem}

The game in Fig. \ref{example:cexample} is not leaktight. We conjecture that for
leaktight games the finite-memory strategies are as powerful as the
infinite-memory ones ($\val(s)=\val^f(s)$). One can prove that for all
distributions on the states of the game there exists an optimal (mixed) action
for the minimizer, and this way construct an optimal strategy. This strategy has
in general unbounded memory. But intuitively for leaktight games the exact
distribution is not important, only the support. We leave the veracity of this
conjecture as an open problem.

%%% Local Variables:
%%% mode: latex
%%% TeX-master: "asymmetric-limitsure"
%%% End:

%%
%% Bibliography
%%

%% Either use bibtex (recommended), but commented out in this sample

\section*{Conclusion}
We have defined a class of stochastic games with partial observation where the maxmin-reachability problem is decidable. This holds under the assumption that maximizer is restricted to deterministic strategies. The extension of this result to the value $1$ problem where maximizer is allowed to use mixed strategies seems rather challenging.

\bibliography{bibliography}
\end{document}